\documentclass[a4paper,USenglish]{lipics-v2019}

\usepackage{xcolor}

\usepackage{framed}

\usepackage{graphicx} 
\usepackage{microtype}  
\usepackage[linesnumbered,ruled,vlined]{algorithm2e}
\usepackage[ruled,vlined]{algorithm2e}
\usepackage{enumitem}
\graphicspath{{./graphics/}}

\usepackage[usestackEOL]{stackengine}

\newcommand{\remove}[1]{}

\definecolor{cadmiumgreen}{rgb}{0.0, 0.42, 0.24}
\definecolor{gainsboro}{rgb}{0.86,0.86,0.86}
\definecolor{dandelion}{rgb}{0.94,0.88,0.19}

\colorlet{Mycolor1}{green!40!orange!60!}  

\definecolor{light-gray}{gray}{0.92}  
\definecolor{dark-gray}{gray}{0.20}  

\definecolor{shadecolor}{gray}{0.90}

\title{
The Contest Game
for Crowdsourcing Reviews}
\titlerunning{The Contest Game for
Crowdsourcing Reviews}

\author{Marios Mavronicolas}{Department of Computer Science, University of Cyprus, Cyprus}{mavronic@ucy.ac.cy}{}{Supported by research funds at the University of Cyprus.}
\author{Paul G.~Spirakis}{Department of Computer Science, University of Liverpool, UK}{p.spirakis@liverpool.ac.uk}
{}
{Supported by the EPSRC grant EP/P02002X/1.}

\authorrunning{M.~Mavronicolas and P.~G.~Spirakis}

\Copyright{Left blank}



\EventEditors{}
\EventNoEds{1}
\EventLongTitle{(STACS 2023)}
\EventShortTitle{STACS 2023}
\EventAcronym{STACS}
\EventYear{2023}
\EventDate{2023}
\EventLocation{}
\EventLogo{}
\SeriesVolume{}
\ArticleNo{}

\newtheorem{open problem}[theorem]{{\bf Open Problem}}

\begin{document}

\nolinenumbers
\hideLIPIcs

\remove{\title{{\bf \textcolor{red}{The Impact of Efforts
on Pure Nash Equilibria in the Contest Game for Crowdsourcing Reviews:
Existence and Computation}}}}

\remove{\author{
{\sl Marios Mavronicolas}\thanks{Department of Computer Science,
University of Cyprus.}
\and
{\sl Paul Spirakis}\thanks{Department of Computer Science,
University of Liverpool.} 
}}


\maketitle

\pagenumbering{arabic}

\begin{abstract}
We consider a {\it contest game}
modelling a contest where
reviews for a {\it proposal} 
are crowdsourced from $n$ 
{\it players}.
Player $i$
has a {\it skill} $s_{i}$, 
strategically chooses a {\em quality} 
$q \in \{ 1, 2, \ldots, Q \}$
for her review
and pays an {\it effort} 
${\sf f}_{q} \geq 0$,
strictly increasing with $q$.
Under {\it voluntary participation,}
a player may opt to not write a review,
paying zero effort;
{\it mandatory participation}
does not provide this option.
For her effort,
she is awarded a 
{\it payment} per her {\it payment function},
which is either {\it player-invariant},
like, e.g., 
the popular {\it proportional allocation}, 
or {\it player-specific};
it is {\it oblivious}
when it does not depend
on the numbers of players
choosing a different quality.
The {\it utility} 
to player $i$ 
is the difference 
between her payment
and her {\it cost,}
calculated by a {\it skill-effort} function 
$\Lambda (s_{i}, {\sf f}_{q})$.
Skills may vary for {\it arbitrary players};
{\it anonymous players} means
$s_{i} = 1$
for all players $i$.
In a {\it pure Nash equilibrium,} 
no player could unilaterally increase her utility
by switching to a different quality.
We show the following results
about the existence and the computation
of a pure Nash equilibrium:
\begin{itemize}

\item
We present 
an exact potential
to show the existence
of a pure Nash equilibrium
for the contest game
with arbitrary players 
and player-invariant and oblivious payments. 
A particular case of this result
provides an answer to an open question 
from~\cite{BKLO22}.
In contrast, 
a pure Nash equilibrium
might not exist
{\it (i)}
for player-invariant payments,
even if players are anonymous,
{\it (ii)}
for proportional allocation payments
and arbitrary players,
and {\it (iii)}
for player-specific payments,
even if players are anonymous;
in the last case,
it is ${\cal NP}$-hard to tell.
These counterexamples
prove the tightness of our existence result.

\item
We show that
the contest game
with proportional allocation,
voluntary participation
and anonymous players
has the {\it Finite Improvement Property}, or {\it FIP};
this yields two pure Nash equilibria.
The {\it FIP} carries over
to mandatory participation,
except that there is now a single
pure Nash equilibrium.
For arbitrary players,
we determine a simple sufficient
condition for the {\it FIP}
in the special case where
the skill-effort function
has the product form
$\Lambda (s_{i}, {\sf f}_{q})
= s_{i}\, {\sf f}_{q}$.

\item
We introduce a novel, discrete concavity property
of player-specific payments,
namely {\it three-discrete-concavity,}
which we exploit to devise,
for constant $Q$,
a polynomial-time $\Theta (n^{Q})$ algorithm
to compute a pure Nash equilibrium
in the contest game with arbitrary players;
it is a special case of a
$\Theta \left( n\, Q^{2}\,
               \binom{\textstyle n+Q-1}{\textstyle Q-1} 
        \right)$
algorithm
for arbitrary $Q$ that we present.
This settles the parameterized complexity of the problem
with respect to the parameter $Q$. 
The computed equilibrium is {\it contiguous}:
players with higher skills
are contiguously assigned to lower qualities.
Both three-discrete-concavity
and the algorithm extend naturally
to player-invariant payments.

\remove{
\item
\textcolor{blue}{
We show that a pure Nash equilibrium
exists
for the contest game with proportional allocation
and anonymous players,
for the product skill-effort function
and a particular scenario of mandatory participation.
We present a constructive proof
via a $\Theta (\max \{ Q, n \})$ algorithm.
Starting with the two highest qualities,
we greedily proceed to the lowest,
focusing each time on a pair of qualities:
maintaining players previously assigned
to higher qualities, 
we split the players
assigned
to the higher of the two
between the two qualities currently considered
so as to enforce equilibrium.}
}

\end{itemize}
\noindent
\remove{These results are complemented with extentions 
in various directions;
for example,
we devise simple $\Theta (1)$ algorithms
under proportional allocation,
taking $\Lambda (s_{i}, {\sf f}_{q})
= s_{i}\, {\sf f}_{q}$
and making stronger assumptions on skills and efforts
for both arbitrary and proposal-indifferent and 
anonymous players.}

\remove{\item
A very simple $\Theta (1)$ algorithm,
which assigns all players to the lowest quality.
The algorithm works for arbitrary players
under the assumption that
skills are lower-bounded:
$\min_{i \in [n]} 
   s_{i}
 \geq
 \frac{\textstyle {\sf f}_{2}}
      {\textstyle {\sf f}_{2} - {\sf f}_{1}}$;
it also works for anonymous players
under the assumption
${\sf f}_{2} - {\sf f}_{1} \geq 1$.

\end{itemize}

\end{itemize} 

\end{itemize}}
\end{abstract}

\newpage

\section{Introduction}
\label{introduction}



{\it Contests}~\cite{V2015} 
are modelled as games
where strategic contestants,
or {\it players},
invest efforts in competitions
to win valuable prizes,
such as monetary awards, 
scientific credit or social reputation.
Such competitions are ubiquitous
in contexts such as 
promotion tournaments in organizations,
allocation of campaign resources,
content curation and selection in online platforms,
financial support of scientific research
by governmental institutions
and
question-and-answer forums.
This work joins an active research thread
on the existence, computation and efficiency
of (pure) Nash equilibria
in games 
for crowdsourcing, content curation,
information aggregation and other relative tasks ~\cite{AHM14,AW22,BNPS21,BGM23,BKLO22,DLLQ22,EG16,EGG21,EGG22,EGG22a,FLZ09,GM11,MCKL14,XQYL14}.

In a {\it crowdsourcing contest} (see, e.g.,~\cite{CHS19,DV09,S20}),
solutions to a certain task are solicited.
When the task is the evaluation of proposals requesting funding,
a set of expert advisors,
or {\it reviewers,}
file peer-reviews of the proposals.
We shall consider a contest game for crowdsourcing reviews,
embracing and wide-extending a corresponding game
from~\cite[Section 2]{BKLO22}
that was motivated by issues 
in the design of blockchains and cryptocurrencies.
In the contest game,
funding agencies wish to collect 
peer-reviews of esteem {\it quality}.
{\it Costs} are incurred to reviewers;
they reflect various overheads, 
such as time, participation cost 
or reputational loss,
and are supposed to increase with the
reviewers' {\it skills} 
and {\it efforts}.\footnote{One might argue that
the cost of a reviewer for writing a review
of a given quality decreases with her skill
and claim that skill is a misnomer;
however,
it can also be argued that skilled players
are incurred higher costs
upon drawing more skills than necessary
for writing a decent review.
For consistency,
we chose to keep using skills
in the same way as in~\cite{BKLO22}.} 
Both skills and efforts are modelled as discrete;
such modelling is natural since,
for example,
monetary expenditure,
the time to spend on projects,
and man-power are usually measured in discrete units.
Naturally,
efforts increase with the achieved qualities of the reviews.  
Efforts map collectively into {\it payments} 
rewarded to the reviewers
to counterbalance their costs.
We proceed to formalize these considerations.

\subsection{The Contest Game for Crowdsourcing Reviews}
\label{the contest game for crowdsourcing reviews}

\noindent
We assume familiarity 
with the basics of finite games,
as articulated, e.g.,
in~\cite{KP17};
we shall restrict attention to finite games.
In the {\it contest game for crowdsourcing reviews,} 
henceforth abbreviated as the {\it contest game,}
there are $n$ {\it players} $1, 2, \ldots, n$,
with $n \geq 2$, 
simultaneously writing 
reviews for a {\it proposal}.
Each player 
$i \in [n]$ 
has a {\it skill} $s_{i} > 0$.
Players are {\it anonymous}
if their skills are the same;
then, take $s_{i} = 1$ for all $i \in [n]$.
Else they are {\it arbitrary}.

The {\it strategy} $q_{i}$ of a player $i \in [n]$
is the {\it quality}
of the review she writes;
she chooses it
from a finite set
$\{ 1, 2, \ldots, Q \}$, with $Q \geq 2$.
For a given {\it quality vector}
${\bf q}
= \langle q_{1}, \ldots, q_{n} \rangle$,
the {\it load} on quality $q$,
denoted as ${\sf N}_{{\bf q}}(q)$,
is the number of players choosing quality $q$;
so
$\sum_{q \in [Q]}
  {\sf N}_{{\bf q}}(q)
 = n$.
A partial quality vector
${\bf q}_{-i}$ results by excluding $q_{i}$
from ${\bf q}$,
for some player $i \in [n]$. 
${\sf Players}_{{\bf q}}(q)$
is the set of players choosing quality $q$
in 
${\bf q}$.
${\sf f}_{q}$ is
the {\it effort} paid by a player 
writing a review of quality $q$; 
it is an increasing function of $q$
with
${\sf f}_1 < {\sf f}_2 < \ldots < {\sf f}_Q$.
{\it Mandatory participation} is modeled by
setting ${\sf f}_{1} > 0$;
under {\it voluntary participation,}
modeled by setting ${\sf f}_{1} = 0$,
a player may choose not to write a review
and save effort.

Given a quality vector ${\bf q}$
and a player $i \in [n]$,
the {\it payment} awarded to player $i \in [n]$
for her review
is the value 
${\sf P}_{i}({\bf q})$
determined by her {\it payment function}
${\sf P}_{i}$,
obeying the {\it normalization condition}
$\sum_{k \in [n]}
   {\sf P}_{k}({\bf q}) 
   \leq 
   1$.       
Payments are 
{\it oblivious} if
for any player $i \in [n]$
and quality vector ${\bf q}$,
${\sf P}_{i}({\bf q})
 = {\sf P}_{i}({\sf N}_{{\bf q}}(q_{i}), {\sf f}_{q_{i}})$;
that is,
${\sf P}_{i}({\bf q})$
depends only on the quality $q_{i}$ chosen by player $i$
and the load on it.
Note that oblivious payments
are not necessarily player-invariant
as for different players $i, k \in [n]$,
it is not necessary that ${\sf P}_{i} = {\sf P}_{k}$.
Payments are 
{\it player-invariant}
if for every quality vector ${\bf q}$,
for any players $i, k \in [n]$
with $q_{i} = q_{k}$,
${\sf P}_{i}({\bf q}) 
 = 
 {\sf P}_{k}({\bf q})$;
thus, players choosing the same quality
are awarded the same payment.
A player-invariant payment function ${\sf P}_{i}({\bf q})$
can be represented by a
two-argument payment function
${\sf P}_{i}(q, {\bf q}_{-i})$,
for a quality $q \in [Q]$
and a partial quality vector ${\bf q}_{-i}$,
for a player $i \in [n]$.
We consider the following player-invariant payments:
\begin{itemize}

\item
The {\it proportional allocation}
${\sf PA}_{i}({\bf q}) 
= \frac{\textstyle {\sf f}_{q_{i}}}    	    
       {\textstyle 
                    \sum_{k \in [n]}
                    {\sf f}_{q_{k}}}$;
thus,
$\sum_{i \in [n]}
  {\sf PA}_{i}({\sf q}) =
 \frac{\textstyle \sum_{i \in [n]}
                   {\sf f}_{q_{i}}}
      {\textstyle 
                  \sum_{i \in [n]}
                   {\sf f}_{q_{i}}} = 1$.
Proportional allocation
is widely studied in the context of contests
with smooth allocation of prizes (cf.~\cite[Section 4.4]{V2015}).
For proportional allocation
with voluntary participation (by which ${\sf f}_{1} = 0$),
in the scenario where all players choose quality 1,
the payment to any player 
becomes $\frac{0}{0}$, 
so it is indeterminate.\footnote{
This means that all values $c$ satisfy
$0 = 0 \cdot c$.}
To remove indeterminacy and make payments well-defined,
we define
the payment to any player choosing quality $1$
in the case where all players choose $1$
to be $0$.
Note that proportional allocation
is not oblivious.

\item
The {\it equal sharing per quality} 
${\sf ES}_{i}({\bf q})
 = 
  {\sf C}_{{\sf ES}}
  \cdot
  \frac{\textstyle {\sf f}_{q_{i}}}
       {\textstyle {\sf N}_{{\bf q}}(q_{i})}$;
so ${\sf f}_{q_{i}}$ is shared evenly
by players choosing $q_{i}$.                         
Since $\sum_{i \in [n]}
        {\sf ES}_{i}
             ({\bf q})
       =
       {\sf C}_{{\sf ES}}
       \cdot
       \sum_{i \in [n]}
       \frac{\textstyle {\sf f}_{q_{i}}}
            {\textstyle {\sf N}_{{\bf q}}(q_{i})}$,
we take
${\sf C}_{{\sf ES}}
 =
 \left(
 \max_{{\bf q}}
   \sum_{i \in [n]}
 \frac{\textstyle {\sf f}_{q_{i}}}
      {\textstyle {\sf N}_{{\bf q}}(q_{i})}
 \right)^{-1}$.
Note that the equal sharing per quality
is different from the standard equal sharing,
by which {\em all} players choosing quality at least some $q \in [Q]$
share ${\sf f}_{q}$ equally.
Thus,
standard equal sharing
is not oblivious,
while the equal sharing per quality is.
Both the equal sharing per quality and the equal sharing
allow for a player's payment to decrease
with an increase in quality;
this happens, for example,
in standard equal sharing
when a player switches from a lower quality with very high load
to a higher quality with a significantly smaller total load
on qualities at least the higher quality.  

\item
The {\it $K{\sf Top}$ allocation}
$K{\sf Top}_{i}({\bf q})
 = 
 {\sf C}_{K{\sf Top}}
 \cdot
 \left\{ \begin{array}{ll}
            0\, , & \mbox{if $q_{i} \leq Q-K$} \\
            \frac{\textstyle {\sf f}_{q_{i}}}
                 {\textstyle {\sf N}_{{\bf q}}(q_{i})}\, ,
           & \mbox{if $q_{i} > Q-K$}
          \end{array}                 
  \right.$;
so players choosing a quality $q$ higher than
a certain quality $Q - K$
share ${\sf f}_{q}$ evenly.  
Since $\sum_{i \in [n]}
        K{\sf Top}_{i}({\bf q}^{\ell})
       =
       {\sf C}_{K{\sf Top}}
       \sum_{q_{i} > Q-K}
         \frac{\textstyle {\sf f}_{q_{i}}}
              {\textstyle {\sf N}_{{\bf q}}(q_{i})}$, 
we take
${\sf C}_{K{\sf Top}}
 =
 \left( 
 \max_{{\bf q}^{\ell}}
   \sum_{q_{i} > Q-K}
     \frac{\textstyle {\sf f}_{q_{i}}}
          {\textstyle {\sf N}_{{\bf q}}(q_{i})}
 \right)^{-1}$.
Note that the $K{\sf Top}$ allocation
is different from
the standard $K{\sf Top}$ allocation,
considered in, e.g.,~\cite{EG16,JCP14,XQYL14},
by which {\it all} players choosing quality higher
than $Q - K$
share ${\sf f}_{q}$ equally;
so the utility of a player $i$ 
choosing a quality $q_{i} > Q-K$
in ${\bf q}$
is
$\frac{\textstyle {\sf f}_{q_{i}}}
      {\textstyle \sum_{q > Q -K}
                    {\sf N}_{{\bf q}}(q)}$.
Thus,
the standard $K{\sf Top}$ allocation is not oblivious,
while the $K{\sf Top}$ allocation is. 

\end{itemize}
\remove{
Listed in~\cite[Section 6.1.3]{V2015}
are more examples of player-invariant payment functions,
including {\it proportional-to-marginal contribution}
(motivated by the marginal contribution condition
in {\it (monotone) valid utility games}~\cite{V02})
and {\it Shapley-Shubick}~\cite{S53,S62}.                       
Proportional allocation
(resp, Equal sharing)
is considered 
in 
the related works~\cite{BGM23,BKLO22}
(resp.,~\cite{EGG22}).}

A generalization of a player-invariant payment function 
results by allowing the payment to player $i \in [n]$
to be a function 
${\sf P}_{i}(i, {\bf q})$
of both $i$ and ${\bf q}$;
it is called a {\it player-specific} payment function.
The {\it cost} or {\it skill-effort function}
${\sf \Lambda}: 
{\mathbb{R}}_{\geq 1} \times {\mathbb{R}}_{\geq 0}
\rightarrow {\mathbb{R}}_{\geq 0}$,
with ${\sf \Lambda} (\cdot,0) = 0$, 
is a monotonically increasing, polynomial-time computable
function
in both skill and effort.
For a quality vector ${\bf q}$,
the {\it utility function} 
is assumed to be of
quasi-linear form with respect to payment and cost
and is defined as
${\sf U}_{i}({\bf q})
= 
    {\sf P}_{i}({\bf q})
    -
    {\sf \Lambda} (s_{i}, 
                   {\sf f}_{q_{i}})$,
for each player $i \in [n]$.
In a {\it pure Nash equilibrium} ${\bf q}$,
for every player $i \in [n]$
and deviation of her 
to strategy
$q \in [Q]$, $q \neq q_{i}$,
${\sf U}_{i}({\bf q}) \geq
 {\sf U}_{i}(q, {\bf q}_{-i})$;
so
no player could increase her utility
by unilaterally
switching to a different quality. 
We consider the following problems
for deciding the existence
of a pure Nash equilibrium
and computing one if there is one:
\begin{itemize}

\item
{\sc $\exists$PNE with Player-Invariant and Oblivious Payments}

\item
{\sc $\exists$PNE with Player-Invariant Payments}

\item
{\sc $\exists$PNE with Proportional Allocation and Arbitrary Players}

\item
{\sc $\exists$PNE with Proportional Allocation and Anonymous Players}

\item
{\sc $\exists$PNE with Player-Specific Payments}

\end{itemize}
\remove{Clearly,
the assumption
$\min_{i \in [n], \ell \in [n]}
   {\sf \Lambda} (s_{i \ell}, {\sf f}_{1}) 
 >
 \frac{\textstyle {\sf T}}
      {\textstyle m}$
excludes the existence
of a pure Nash equilibrium.
So we assume that
$\min_{i \in [n], \ell \in [n]}
   {\sf \Lambda} (s_{i \ell}, {\sf f}_{1}) 
 \leq
 \frac{\textstyle {\sf T}}
      {\textstyle m}$
in order to render
{\sc PNE in Contest Game}
non-trivial.}

\noindent
The most significant difference between
the contest game and
the contest games traditionally considered 
in Contest Theory~\cite{V2015}
is that the it adopts 
players with a {\it discrete} action space,
choosing over a finite number of qualities,
while the latter focus on players 
with a {\it continuous} one.
(See~\cite{CS07} for an exception.)
Alas,
the contest game
is comparable to classes of contests
studied in Contest Theory~\cite{V2015}
with respect to several characteristics:
\begin{itemize}

\item
Casting qualities as individual contests,
the contest game 
resembles
{\it simultaneous contests} 
(cf.~\cite[Section 5]{V2015}),
in which players simultaneously invest efforts
across the set of contests.

\item 
While in an {\it all-pay contest} 
(cf.~\cite[Chapter 2]{V2015})
all players competing for a non-splittable prize
must pay for their bid
and the winner takes all of it,
all players are awarded payments,
summing up to at most 1,
in the contest game.

\item
The utility 
${\sf U}_{i}({\bf q}) = {\sf P}_{i}({\bf q}) 
- {\sf \Lambda}(s_{i}, {\sf f}_{q_{i}})$
in the contest game
can be cast as {\it smooth} 
(cf.~\cite[Chapter 4]{V2015}):
{\it (i)} 
each player receives a portion 
${\sf P}_{i}({\bf q})$ of the prize
according to an allocation mechanism
that is a smooth function of the invested efforts
$\{ {\sf f}_{q} \}_{q \in [Q]}$
(except when all players invest zero effort 
(cf.~\cite[start of Section 4]{V2015},
which may happen under proportional allocation
with voluntary participation)
and {\it (ii)}
utilities are quasilinear in payment and cost;
in this respect,
${\sf U}_{i}$ corresponds to a 
{\it contest success function}~\cite{S96}.

\end{itemize}

\begin{shaded}
We shall need some
definitions from Game Theory,
applying to finite games with players $i$
maximizing utility ${\sf U}_{i}$.
All types of potentials map profiles to numbers.
A game is an {\it (exact) potential game}~\cite{MS96}
if it admits a {\it exact potential} $\Phi$:
for each player $i \in [n]$,
for any pair $q_{i}$ and $q'_{i}$ of her strategies
and for any partial profile ${\bf q}_{-i}$,
${\sf U}_{i}(q'_{i}, {\bf q}_{-i}) - 
 {\sf U}_{i}(q_{i}, {\bf q}_{-i})
 =
 \Phi(q'_{i}, {\bf q}_{-i}) 
 - 
 \Phi(q_{i}, {\bf q}_{-i})$.
A game is an {\it ordinal potential game}~\cite{MS96}
if it admits a {\it ordinal potential} $\Phi$:
for each player $i \in [n]$,
for any pair $q_{i}$ and $q'_{i}$ of her strategies
and for any partial profile ${\bf q}_{-i}$,
${\sf U}_{i}(q'_{i}, {\bf q}_{-i}) > 
 {\sf U}_{i}(q_{i}, {\bf q}_{-i})$
if and only if 
$\Phi(q'_{i}, {\bf q}_{-i}) >
 \Phi(q_{i}, {\bf q}_{-i})$.
A game is a {\it generalized ordinal potential game}~\cite{MS96}
if it admits a {\it generalized ordinal potential} $\Phi$:
for each player $i \in [n]$,
for any pair $q_i$ and $q'_{i}$ of her strategies,
and for any partial profile ${\bf q}_{-i}$,
${\sf U}_{i}(q_{i}, {\bf q}_{-i}) 
 > 
 {\sf U}_{i}(q_{i}', {\bf q}_{-i})$
implies
$\Phi (q_{i}, {\bf q}_{-i})
 >
 \Phi (q_{i}', {\bf q}_{-i})$.
So a potential game
is a strengthening of an ordinal potential game,
which is a strengthening of a 
generalized ordinal potential game.
Every generalized ordinal potential game
has at least one pure Nash equilibrium~\cite[Corollary 2.2]{MS96}.

We recast some definitions from Game Theory
in the context of the contest game.
An {\it improvement step} 
out of the quality vector ${\bf q}$ and 
into the ${\bf q'}$
occurs when there is a unique player $i \in [n]$
with
$q_{i} \neq q_{i}'$
such that
${\sf U}_{i}({\bf q})
 <
 {\sf U}_{i}({\bf q'})$;
so it is profitable for player $i$ to switch
from $q_{i}$ to $q_{i}'$.
An {\it improvement path}
is a sequence 
${\bf q}^{(1)},
 {\bf q}^{(2)}, \ldots, 
 $
such that for each quality vector
${\bf q}^{(\rho)}$ in the sequence,
where $\rho \geq 1$,
there occurs
an improvement step out of ${\bf q}^{\rho}$
and into ${\bf q}^{(\rho + 1)}$.
A {\it finite} improvement path   
has finite length.
The {\it Finite Improvement Property,} 
abbreviated as {\it FIP,}
requires that
all improvement paths
are finite;
that is,
there are no cycles
in the directed {\it quality improvement graph,} 
whose vertices are the quality vectors
and there is an edge from quality vector
${\bf q}^{(1)}$ to ${\bf q}^{(2)}$
if and only if
an improvement step occurs from ${\bf q}^{(1)}$
to ${\bf q}^{(2)}$.
Every game with the {\it FIP} has
a pure Nash equilibrium:
a {\it sink} in the quality improvement graph; 
there are games without the {\it FIP} that also have~\cite{MS96}.
By~\cite[Lemma 2.5]{MS96},
a game has a generalized ordinal potential
if and only if
it has the {\it FIP}.

\end{shaded}

\subsection{Results}

We study
the existence and the computation of
pure Nash equilibria for the contest game.
{\it When do pure Nash equilibria exist for 
arbitrary players,
player-invariant or player-specific payments
and for arbitrary $n$
and $Q$?}
For the special case of the contest game
with proportional allocation payments
and a skill-effort function
$\Lambda (s_{i}, {\sf f}_{q})
 =
 s_{i} {\sf f}_{q}$,
this has been advocated
as a significant open problem
in~\cite[Section 6]{BKLO22}.
{\it What is the time complexity of
deciding the existence of a pure Nash equilibrium
and computing one in case there exists one?}
{\it Is this complexity affected by
properties of the payment or the skill-effort function,
or by numerical properties of skills and efforts, and how?}
We shall present three major results:
\begin{itemize}

\item
Every contest game 
with arbitrary players 
and player-invariant and oblivious payments
has a pure Nash equilibrium,
for any values of $n$
and $Q$
and any skill-effort function ${\sf \Lambda}$
(Theorem~\ref{pure existence}).
We devise an {\it exact potential}~\cite{MS96}
for the contest game
and resort to the fact that every {\it potential game}
has a pure Nash equilibrium~\cite[Corollary 2.2]{MS96}. 
By Theorem~\ref{pure existence},
the contest game
with equal sharing per quality
and {\sf $K$Top} allocation
has a pure Nash equilibrium.
However,
existence does not extend beyond player-invariant
{\em and} oblivious payments:
We prove the tightness of our existence result
(Theorem~\ref{pure existence})
by exhibiting simple contest games
with no pure Nash equilibrium
when:
\begin{itemize}

\item
Payments are player-invariant but not oblivious,
even if players are anonymous
(Proposition~\ref{icalp counterexample 1}).

\item
Payments are proportionally allocated
and players are arbitrary
(Proposition~\ref{icalp counterexample 2}).

\item
Payments are player-specific,
even if players are anonymous
(Proposition~\ref{counter example}).
The ${\mathcal{NP}}$-completeness of
deciding the existence of a pure Nash equilibrium
follows by a simple reduction from the problem
of deciding the existence of a pure Nash equilibrium
in a succinctly represented 
strategic game~\cite[Theorem 2.4.1]{S74}
(Theorem~\ref{np completeness}).

\end{itemize}

\item
We show that the contest game
with proportional allocation,
voluntary participation
and anonymous players
has the {\it FIP}
(Theorem~\ref{sto paderborn}).
The contest game is found to have
two pure Nash equilibria
in this case.
A simplification of the proof
for voluntary participation establishes the {\it FIP}
for mandatory participation
(Theorem~\ref{wow paderborn});
the number of pure Nash equilibria drops to one.
As the key to establish these results,
we show the {\it No Switch from Lower Quality to Higher Quality} Lemma:
in an improvement step,
a player necessarily switches from a higher quality
to a lower quality
(Lemma~\ref{kyriakatiko}).

These results are complemented
with a very simple, $\Theta (1)$ 
algorithm that works under proportional allocation,
for arbitrary players,
with ${\sf \Lambda} (s_{i}, {\sf f}_{q}) 
= s_{i}\, {\sf f}_{q}$
and making stronger assumptions on skills and efforts
to compute a pure Nash equilibrium
(Theorem~\ref{natasa}).
The algorithm simply assigns all players to quality $1$;
so it runs in optimal time ${\sf \Theta}(1)$.

\item
Finally, we consider 
a player-specific payment function
that is also 
{\it three-discrete-concave}:
for any triple of qualities $q_{i}$, $q_{k}$ and $q$,
the difference between the payments 
when incrementing the load on $q$
and decrementing the load on $q_{i}$
is at most the difference between the payments
when incrementing the load on $q_{k}$
and decrementing the load on $q$.
Three-discrete-concave functions
make a new class of discrete-concave functions
that we introduce;
similar classes of discrete-concave functions,
such as {\it $L$-concave},
are extensively discussed in the excellent monograph
by Murota~\cite{M03}.
We present a
$\Theta \left( n
               \cdot
               Q^{2}\,
               \binom{\textstyle n+Q-1}{\textstyle Q-1} 
        \right)$
algorithm
to decide the existence of and compute
a pure Nash equilibrium
for 
three-discrete-concave player-specific payments
and arbitrary players
(Theorem~\ref{the most arbitrary}).

Exhaustive enumeration of {\em all} quality vector
incurs an {\em exponential} $\Theta (Q^{n})$ time complexity.
To bypass the intractability,
we focus on {\em contiguous} profiles,
where any players $i$ and $k$,
with $s_{i} \geq s_{k}$,
are assigned to qualities $q$ and $q'$,
respectively,
with $q \leq q'$;
they offer a significant advantage:
the cost for their exhaustive enumeration
drops to 
$\Theta \left( \binom{\textstyle n+Q-1}{\textstyle Q-1}
        \right)$.
We prove the {\it Contigufication Lemma}: 
any pure Nash equilibrium
for the contest game
can be transformed into a contiguous one
(Proposition~\ref{transformation into contiguous}).
So,
it suffices to search for a 
contiguous, pure Nash equilibrium.
The algorithm is polynomial-time $\Theta (n^{Q})$
for {\it constant} $Q$,
settling the parameterised complexity of the problem
when payments are
player-specific.

We extend the algorithm for three-discrete-concave
player-specific payments
to obtain a
$\Theta \left( \max
               \{ n,
               Q^{2}
               \}
               \cdot 
               \binom{\textstyle n+Q-1}{\textstyle Q-1} 
        \right)$
algorithm
for three-discrete-concave player-invariant payments
(Theorem~\ref{the most most arbitrary}).
The improved time complexity
for arbitrary $Q$
in comparison to the case of three-discrete-concave 
player-specific payments
is due to the fact that
the player-invariant property
allows dealing with the payment of only one,
instead of all,
of the players
choosing the same quality.

\end{itemize}

\subsection{Related Work and Comparison}

The contest game studied here 
is inspired by, embraces and extends
in two significant ways
an interesting contest game introduced in~\cite{BKLO22}.
First, we consider an arbitrary payment function,
whereas~\cite{BKLO22} focuses on proportional allocation.
Second,
we consider a cost function 
that is an arbitrary function of skill and effort,
whereas~\cite{BKLO22} focuses on
the product of skill and effort.
Although we have considered
a single proposal in our contest game,
multiple proposals can also be accommodated, as in~\cite{BKLO22}.

Casting qualities as {\it resources,}
the contest game resembles {\it unweighted congestion games}~\cite{R73};
adopting their original definition in~\cite{R73},
there are, though, two significant differences:
{\it (i)} players choose sets of resources in a 
(weighted or unweighted) congestion game
while they choose a single quality in a contest game,
and {\it (ii)}
the utilities (specifically, their payment part)
depend on the loads on {\it all} qualities
in a contest game,
while costs on a resource depend only
on the load on the resource
in an congestion game.
However,
their dissimilarity is trimmed down
when restricting the comparison
to contest games with an oblivious payment function,
where a payment depends only on
the load on the particular quality, 
and to {\it singleton} (unweighted) congestion games,
first introduced in~\cite{QS94},
where each player chooses a single resource.
Note that
the payments in a contest game with an oblivious payment function
may be player-specific,
while, in general, costs in a singleton congestion game are not.

Congestion games with player-specific payoffs
were introduced by Milchtaich~\cite{M96}
as singleton congestion games
where the payoff to a player
choosing a resource is given by a player-specific payoff function.
(In fact,
player-specific payments in this paper
have been inspired by player-specific payoffs in~\cite{M96}.)
In~\cite[Theorem 2]{M96},
it is shown that,
under a standard monotonicity assumption on the payoff function,
these games 
always have a pure Nash equilibrium.
An example is provided in~\cite[Section 5]{M96}
of a congestion game with player-specific payoffs
that lacks the {\it Finite Improvement Property (FIP)}.
In contrast,
Theorem~\ref{pure existence} shows that
the contest game with a player-invariant 
and oblivious payment function,
a special case
of a congestion game with player-specific payoffs,
has a potential function;
thus,
it identifies a subclass of congestion games 
with player-specific payoffs
that does have the stronger {\it FIP}.

\remove{\textcolor{blue}{
Defined by Milchtaich~\cite[Section 3]{M96},
a {\it symmetric} congestion game with player-specific payoffs
on $r$ resources
has each player sharing the same set of payoff functions
$S_{1}, S_{2}, \ldots, S_{r}$;
thus,
players choosing the same resource
have the same payoff.
In contrast,
in a contest game with a player-invariant payment function,
players choosing the same quality
need not get the same utility
as the utility for a player $i$
includes also the skill-effort cost,
which depends on the skill $s_{i}$;
they get the same utility only if players are anonymous.}}

Gairing {\em et al.}~\cite{GMT11}
consider cost-minimizing players
and {\it non-singleton} congestion games 
with player-specific costs;
~\cite[Theorem 3.1]{GMT11},
shows that there is a potential 
for the strict subclass of congestion games
with {\it linear} player-specific costs 
of the form
${\sf f}_{ie}(\delta) = \alpha_{ie} \cdot \delta$,
where $\alpha_{ie} \geq 0$,
for a player $i$ and a resource $e$;
$\delta$ is the number of players choosing resource $e$.
For the potential function result (Theorem~\ref{pure existence})
for the contest game 
with a player-invariant and oblivious payment function,
we consider {\em general} player-specific utilities
of the form
${\sf U}_{i}(q) = 
{\sf P}_{i}(i, {\sf N}(q)) - {\sf \Lambda}(s_{i}, {\sf f}_{q})$,
where ${\sf P}_{i}(i, {\sf N}(q)) \geq 0$
is not necessarily linear
and ${\sf \Lambda}$ is an arbitrary non-negative function,
which is independent of ${\sf N}(q)$
and could be non-monotone.
Theorem~\ref{pure existence} is 
a significant generalization of~\cite[Theorem 3.1]{GMT11},
which assumed {\it linear} player-specific costs,
and an extention of it,
due to the subtracted term
${\sf \Lambda}(s_{i}, {\sf f}_{q})$.
however,
it is also a restriction of ~\cite[Theorem 3.1]{GMT11}, 
since the contest game is singleton
and ${\sf P}_{i}$
is assumed player-invariant.

The contest games considered in the proofs
of the existence of pure Nash equilibria
for ~\cite[Theorems 1 and 3]{BKLO22} 
assume $Q=3$ and $Q=2$,
respectively,
and deal with proportional allocation, 
voluntary participation
and a skill-effort function
${\sf \Lambda}(s_{i}, {\sf f}_{q})
 =
 s_{i} {\sf f}_{q}$,
for any player $i \in [n]$ and quality $q \in [Q]$. 
Pure Nash equilibria are ill-defined
in all considered cases of voluntary participation
as they ignore
the indeterminacy arising in case
all players choose quality 1.
Putting aside this correctness issue,
Theorem~\ref{pure existence}
generalizes the context of~\cite[Theorem 3]{BKLO22} 
from the case $Q=2$ to arbitrary $Q$,
for {\it any} player-invariant and oblivious payment function
and {\it any} skill-effort function;
Theorems~\ref{sto paderborn} and~\ref{wow paderborn} 
generalize the context of ~\cite[Theorem 1]{BKLO22}
from $Q=3$ to arbitrary $Q$,
while they significantly strengthen the claimed results
for these ill-defined cases,
since 
{\it (i)}
they establish the {\it FIP,}
which is a property stronger than the existence of
a pure Nash equilibrium,
{\it (ii)}
they cover together both voluntary and mandatory participation,
and {\it (iii)}
they explicitly determine the pure Nash equilibria and their number,
while the outlined convergence arguments for claiming
~\cite[Theorem 1]{BKLO22} do not.

The contest game is related to {\em project games}~\cite{BGM23},
where each {\it weighted} player $i$ selects 
a single {\it project} $\sigma_{i} \in S_{i}$ 
among those available to him, 
where several players may select the same project.
Weights $w_{i, \sigma_{i}}$ are {\it project-specific};
they are called {\em universal}
when they are fixed for the same project
and {\it identical} when the fixed weights
are the same over all projects.
The utility of player $i$ is a fraction $r_{\sigma_{i}}$
of the proportional allocation of weights
on the project $\sigma_{i}$. 
Projects 
can be considered to correspond 
to qualities 
in the contest game,
which,
in contrast,
has, in general,
neither weights nor fractions
but has the extra term ${\sf \Lambda} (s_{i}, {\sf f}_{q})$
for the cost.

For the contest game in~\cite{EGG22},
there are $m$ {\it activities}
and player $i \in [n]$ 
chooses an {\it output vector}
${\bf b}_{i} = \langle b_{i1}, \ldots, b_{im} \rangle$,
with $b_{i \ell} \in {\mathbb{R}}_{\geq 0}$, $\ell \in [m]$;
the case $b_{i \ell} = 0$
corresponds to voluntary participation.
In contrast,
there are no activities
in the contest game;
but
one may view the single proposal and quality vectors
in it
(as well as in the contest game in~\cite{BKLO22})
as an activity and output vectors, respectively.
There are $C \geq 1$ {\it contests}
awarding prizes to the players based on their output vectors;
allocation
is equal sharing in~\cite{EGG22},
by which players receiving a prize share
are "filtered" using a function $f_{c}$ associated with contest $c$.
The special case of the contest game in~\cite{EGG22}
with $C = 1$
can be seen to correspond to a
contest game in our context;
nevertheless,
to the best of our understanding,
no results transfer between the contest games
in~\cite{EGG22} and in this paper,
as their definitions are different;
for example,
we do not see how to embed output vectors 
in our contest game,
or skill-effort costs
in the contest game in~\cite{EGG22}.

Listed in~\cite[Section 6.1.3]{V2015}
are more examples of player-invariant payments,
including {\it proportional-to-marginal contribution}
(motivated by the marginal contribution condition
in {\it (monotone) valid utility games}~\cite{V02})
and {\it Shapley-Shubick}~\cite{S53,S62}.                       
Games employing
proportional allocation,
equal sharing and $K$-{\sf Top} allocation
have been studied,
for example,
in~\cite{BGM23,C16,FLZ09,PPV14,Z05},
in~\cite{EGG22,MCKL14}
and in~\cite{EG16,JCP14,XQYL14},
respectively.
Accounts on
proportional allocation and equal sharing
in simultaneous contests
appear in~\cite[Section 5.4 \& Section 5.5]{V2015},
respectively.
Player-invariant payments
enhance
{\it Anonymous Independent Reward Schemes (AIRS)}~\cite{CTWXY22},
where payments, 
termed as {\it rewards,}
are only allowed to depend
on the quality of the individual review,
or {\it content}
in the context of user-generated content platforms.

A plethora of results in Contest Theory
establish the inexistence of pure Nash equilibria
in contests with continuous strategy spaces;
see, e.g,~\cite{BKdV96} or~\cite[Example 1.1]{S20}.
Still for continuous strategy spaces,
for proportional allocation,
existence, uniqueness and characterization of pure Nash equilibria
is established in~\cite[Theorem 4.9]{V2015}
for two-player contests
and in~\cite{JT04}
for contests with an arbitrary number of players,
assuming additional conditions
on the utility functions.
All-pay contests with discrete action spaces
were considered in~\cite{CS07}. 
In our view,
the analysis of contest games
with discrete action spaces
is more challenging;
it requires combinatorial arguments,
instead of concavity and continuity arguments,
typically employed for contests with continuous action spaces.  

\section{(In)Existence of a Pure Nash Equilibrium}
\label{many proposals}

We show:

\begin{theorem}
\label{pure existence}
The contest game with arbitrary players
and player-invariant and oblivious payments
has an exact potential
and a pure Nash equilibrium.
\end{theorem}

\begin{proof}
Define the function
$\Phi : \{ {\bf q} \} \rightarrow {\mathbb{R}}$ 
as
\begin{eqnarray*}
        \Phi({\bf q}) 
& = & \sum_{q \in [Q]}   
        {\sf \Gamma}({\sf N}_{{\bf q}}(q)) 
        -
        \sum_{k \in [n]}
        \Lambda (s_{k}, 
                 {\sf f}_{q_{k}})\, ,
\end{eqnarray*}
where the function
${\sf \Gamma}: 
 {\mathbb{N}} \cup \{ 0 \}
 \rightarrow 
 {\mathbb{R}}$  
will be defined later.  
We prove that $\Phi$ is an exact potential.

Consider a player $i \in [n]$
switching from strategy $q_{i}$,
to strategy $\widehat{q}_{i}$,
while other players do not change strategies.
So the quality vector
${\bf q}
 = \langle 
   q_{1}, \ldots, 
   q_{(i-1)}, q_{i}, q_{i+1}, 
   \ldots, q_{n}
   \rangle$
is transformed into
$\widehat{{\bf q}} :=  
\langle q_{1}, 
        \ldots, 
        q_{i-1}, 
        \widehat{q}_{i}, 
        q_{i+1}, 
        \ldots, 
        q_{n} 
\rangle$;
thus,
${\sf N}_{\widehat{{\bf q}}}(q_{i})
 =
 {\sf N}_{{\bf q}}(q_{i}) - 1$,
${\sf N}_{\widehat{{\bf q}}}(\widehat{q}_{i})
 =
 {\sf N}_{{\bf q}}(\widehat{q}_{i}) + 1$
and
${\sf N}_{\widehat{{\bf q}}}(\widetilde{q})
 =
 {\sf N}_{{\bf q}}(\widetilde{q})$
for each quality
$\widetilde{q} \neq q_{i}, \widehat{q}_{i}$.
To simplify notation,
denote $q_{i}$ and $\widehat{q}_{i}$ as
$q$ and $\widehat{q}$,
respectively.
So,
{\small
\begin{eqnarray*}
      {\sf U}_{i}({\bf q})
      -
      {\sf U}_{i}(\widehat{{\bf q}})
& = & \left[ {\sf P}_{i}({\bf q})
      \right]	_{\left[ {\sf N}_{{\bf q}}(q),
                  {\sf N}_{{\bf q}}(\widehat{q})
               \right]}
      -
    \left[ {\sf P}_{i}(\widehat{{\bf q}})
    \right]_{\left[ 
              {\sf N}_{{\bf q}}(q) - 1,
              {\sf N}_{{\bf q}}(\widehat{q}) + 1
             \right]}
      + 
        \Lambda (s_{i}, {\sf f}_{\widehat{q}})
        -
        \Lambda (s_{i}, {\sf f}_{q})\, ,
\end{eqnarray*}
}
where 
$\left[ {\sf P}_{i}({\bf q})
 \right]_{\left[ 
           {\sf N}_{{\bf q}}(q),
           {\sf N}_{{\bf q}}(\widehat{q})
          \right]}$
and     
$\left[ {\sf P}_{i}({\bf q})
 \right]_{\left[ 
          {\sf N}_{{\bf q}}(q)-1,
          {\sf N}_{{\bf q}}(\widehat{q})+1
         \right]}$
denote the payments awarded to $i$
when the loads on qualities $q$ and $\widehat{q}$
are 
$\left( {\sf N}_{{\bf q}}(q),
        {\sf N}_{{\bf q}}(\widehat{q})
 \right)$
and
$\left( {\sf N}_{{\bf q}}(q) - 1,
        {\sf N}_{{\bf q}}(\widehat{q}) + 1
 \right)$,
respectively,
while loads on other qualities
remain unchanged. 
So
$\left[ {\sf P}_{i}({\bf q})
 \right]_{\left[ 
           {\sf N}_{{\bf q}}(q),
           {\sf N}_{{\bf q}}(\widehat{q})
         \right]}
 =
 {\sf P}_{i}({\bf q})$
and
$\left[ {\sf P}_{i}({\bf q})
 \right]_{\left[ 
           {\sf N}_{{\bf q}}(q) - 1,
           {\sf N}_{{\bf q}}(\widehat{q}) + 1
         \right]}
 =
 {\sf P}_{i}(\widehat{{\bf q}})$.
Clearly,
{\small
\begin{eqnarray*}
      \Phi({\bf q}) - \Phi(\widehat{{\bf q}})   
& = & {\sf \Gamma}({\sf N}_{{\bf q}}(q)) 
      +
      {\sf \Gamma}({\sf N}_{{\bf q}}(\widehat{q}))
      -        
      \Lambda (s_{i}, {\sf f}_{q})                                
- 
      \left( {\sf \Gamma}({\sf N}_{{\bf q}}(q) -1) 
             + 
             {\sf \Gamma}
             ({\sf N}_{{\bf q}}(\widehat{q}) + 1)
             - 
             \Lambda (s_{i}, {\sf f}_{\widehat{q}}) 
      \right)                                             \\    
& = & 
      {\sf \Gamma}({\sf N}_{{\bf q}}(q))
      - 
      {\sf \Gamma}({\sf N}_{{\bf q}}(q)-1)
- 
      \left(
      {\sf \Gamma}({\sf N}_{{\bf q}}(\widehat{q}) +1)
      - 
      {\sf \Gamma}({\sf N}_{{\bf q}}(\widehat{q}))
      \right)
      + 
         \Lambda (s_{i}, {\sf f}_{\widehat{q}})
         - 
         \Lambda (s_{i}, {\sf f}_{q})\, .                                  
\end{eqnarray*}
}
Now define the function ${\sf \Gamma}$
such that
for a quality vector ${\bf q}$,
for each quality $q \in [Q]$,
{\small
\begin{eqnarray*}
      {\sf \Gamma}({\sf N}_{{\bf q}}(q))
      - 
      {\sf \Gamma}({\sf N}_{{\bf q}}(q)-1)
& = & \left[ {\sf P}_{i}({\bf q})
      \right]_{\left[ {\sf N}_{{\bf q}}(q),
                      {\sf N}_{{\bf q}}(\widehat{q})
              \right]}\, ,
\end{eqnarray*}
}
We set $\widehat{q}$ for $q$
and ${\sf N}_{{\bf q}}(\widehat{q}) + 1$
for ${\sf N}_{{\bf q}}(q)$
to obtain
{\small
\begin{eqnarray*}
      {\sf \Gamma}({\sf N}_{\widehat{{\bf q}}}(\widehat{q}) + 1)       
      - 
      {\sf \Gamma}({\sf N}_{\widehat{{\bf q}}}(\widehat{q}))
& = & \left[ {\sf P}_{i}(\widehat{{\bf q}})
      \right]_{\left[ {\sf N}_{{\bf q}}(q) -1,
                      {\sf N}_{{\bf q}}(\widehat{q}) + 1
              \right]}\, ,             
\end{eqnarray*}
}
\remove{{\small
\begin{eqnarray*}}
      {\sf \Gamma}_{\ell}({\sf N}_{{\bf q}^{\ell}}(q))
      {\sf H}_{{\sf N}_{{\bf q}^{\ell}}(q)}       
      - 
      {\sf \Gamma}_{\ell}({\sf N}_{{\bf q}^{\ell}}(q)-1)
      {\sf H}_{{\sf N}_{{\bf q}^{\ell}}(q) - 1}
& = & \left[ {\sf P}_{i\ell}({\bf q}^{\ell})
      \right]_{\left[ {\sf N}_{{\bf q}^{\ell}}(q),
                      {\sf N}_{{\bf q}^{\ell}}(\widehat{q})
              \right]}
      \remove{-
      \left[ {\sf P}_{i\ell}({\bf q}^{\ell})
      \right]_{\left[ {\sf N}_{{\bf q}^{\ell}}(q) -1,
                       {\sf N}_{{\bf q}^{\ell}}(\widehat{q})+1
               \right]}\,} .              
\end{eqnarray*}
} 
if ${\sf N}_{{\bf q}}(q) \geq 1$,
and ${\sf \Gamma}(0) = 0$. 
Note that ${\sf \Gamma}$
is well-defined:
the left-hand side is a function of ${\sf N}_{{\bf q}}$ only,
as also is the right-hand side
since
${\sf P}_{i}({\bf q})$
is independent of 
{\it (i)}
$i$, since ${\sf P}$ is player-invariant,
and {\it (ii)}
the loads on qualities 
other than $q$, 
since ${\sf P}$ is oblivious.
An explicit formula for 
${\sf \Gamma}({\sf N}_{{\bf q}}(q)$
follows from its definition: 
\remove{Finally, note that
\textcolor{red}{
\begin{eqnarray*}
{\sf \Gamma}_{\ell}(1)
& = &
 \left[ {\sf P}_{i\ell}({\bf q}^{\ell})
 \right]_{\left[ 
          1, {\sf N}_{{\bf q}^{\ell}}(\widehat{q})
         \right]}
 -
 \left[ {\sf P}_{i\ell}({\bf q}^{\ell})
 \right]_{\left[ 
          0, {\sf N}_{{\bf q}^{\ell}}(\widehat{q})
             + 1
         \right]}\, .
\end{eqnarray*}         
}  
}
{\small
\begin{eqnarray*}
{\sf \Gamma}({\sf N}_{{\bf q}}(q)) 
\remove{\frac{\textstyle 
            {\sf H}_{{\sf N}_{{\bf q}}(q) - 1}}
           {\textstyle 
            {\sf H}_{{\sf N}_{{\bf q}}(q)}}
     {\sf \Gamma}({\sf N}_{{\bf q}}(q) - 1)
     +
     \frac{\textstyle \left[ {\sf P}_{i}({\bf q})
                      \right]_{\left[ {\sf N}_{{\bf q}}(q),
                                      {\sf N}_{{\bf q}}(\widehat{q})
                      \right]}}
         {\textstyle {\sf H}_{{\sf N}_{{\bf q}^{\ell}}(q)}} 
     \remove{\textcolor{red}{-
      \frac{\textstyle \left[ {\sf P}_{i\ell}({\bf q}^{\ell})
                      \right]_{\left[ {\sf N}_{{\bf q}^{\ell}}(q) - 1
                                      {\sf N}_{{\bf q}^{\ell}}(\widehat{q}) + 1
                      \right]}}
         {\textstyle {\sf H}_{{\sf N}_{{\bf q}^{\ell}}(q)}}}}
     \\}                      
& = & 
     \left( 
           {\sf \Gamma}({\sf N}_{{\bf q}}(q)-2)
           +
                 \left[ {\sf P}_{i}({\bf q})
                 \right]_{\left[ {\sf N}_{{\bf q}}(q)-1,
                                 {\sf N}_{{\bf q}}(\widehat{q}) + 1
                          \right]}
     \right)
     +
           \left[ {\sf P}_{i}({\bf q})
           \right]_{\left[ {\sf N}_{{\bf q}}(q),
                           {\sf N}_{{\bf q}}(\widehat{q})
                   \right]}\ \
=\ \ \ldots                                                            \\
& = & 
      \left[ {\sf P}_{i}({\bf q})
            \right]_{\left[ 
                     1, {\sf N}_{{\bf q}}(q)
                        +
                        {\sf N}_{{\bf q}}(\widehat{q})
                        - 1
                    \right]}                
      +
            \left[ {\sf P}_{i}({\bf q})
            \right]_{\left[ 
                     2, {\sf N}_{{\bf q}}(q)
                        +
                        {\sf N}_{{\bf q}}(\widehat{q})
                        - 2
                    \right]}
            +
            \ldots
           +    
           \left[ {\sf P}_{i}({\bf q})
           \right]_{\left[ {\sf N}_{{\bf q}}(q),
                           {\sf N}_{{\bf q}}(\widehat{q})
                   \right]} 
\end{eqnarray*}
}
Hence,
by definition of ${\sf \Gamma}$,
{\small
\begin{eqnarray*}
      \Phi({\bf q}) - \Phi(\widehat{{\bf q}})   
& = & \left[ {\sf P}_{i}({\bf q})
      \right]	_{\left[ {\sf N}_{{\bf q}}(q),
                  {\sf N}_{{\bf q}}(\widehat{q})
               \right]}
      -
    \left[ {\sf P}_{i}(\widehat{{\bf q}})
    \right]_{\left[ {\sf N}_{{\bf q}}(q) - 1,
                    {\sf N}_{{\bf q}}(\widehat{q}) + 1
             \right]}
      + \Lambda (s_{i},
                        {\sf f}_{\widehat{q}})
      -
      \Lambda (s_{i},
                        {\sf f}_{q})\, .
\end{eqnarray*}
}
Hence,
$\Phi({\bf q}) - \Phi(\widehat{{\bf q}})   
 = {\sf U}_{i}({\bf q})
   -
   {\sf U}_{i}(\widehat{{\bf q}})$,
$\Phi$ is an exact potential 
and
a pure Nash equilibrium exists.
\end{proof}

\noindent
Since
${\sf \Gamma}$, ${\sf P}$ and ${\sf \Lambda}$
are poly-time computable,
so is also the exact potential $\Phi$ used 
for the proof of Theorem~\ref{pure existence}
since it involves summations of values
of ${\sf \Gamma}$, ${\sf P}$ and ${\sf \Lambda}$.
Hence,
{\sc $\exists$PNE with Player-Invariant and Oblivious Payments} 
$\in {\mathcal{PLS}}$.

\begin{framed}
\begin{open problem}
Determine the precise complexity of 
{\sc $\exists$PNE with Player-Invariant and Oblivious Payments}.
We remark that no ${\mathcal{PLS}}$-hardness results
for computing pure Nash equilibria
are known for either singleton congestion games~\cite{M96}
or for project games~\cite{BGM23}, 
which, in some sense,
are also singleton as the contest game is;
moreover,
all known ${\mathcal{PLS}}$-hardness results for computing pure Nash equilibria in congestion games apply to congestion games
that are not singleton.
These remarks appear to speak against
${\mathcal{PLS}}$-hardness.
\end{open problem}
\end{framed}

\noindent
We next show that existence of pure Nash equilibria
is not guaranteed 
if ${\sf P}$ is not
player-invariant and oblivious simultaneously.
We start by showing:

\begin{proposition}
\label{icalp counterexample 1}
There is a contest game
with mandatory participation,
player-invariant payments
and anonymous players
that has neither the {\it FIP}
nor a pure Nash equilibrium.
\end{proposition}

\begin{proof}
Consider the contest game
with two players 1 and 2 with skill $\frac{\textstyle 1}
                                          {\textstyle 3}$
and three qualities 1, 2 and 3,
with ${\sf f}_{q} = q$
for $q \in [3]$.
So participation is mandatory.
Assume a product skill-effort function
${\sf \Lambda}(\frac{\textstyle 1}
                    {\textstyle 3}, {\sf f}_{q})
 =
 \frac{\textstyle 1}
      {\textstyle 3}
 {\sf f}_{q}$,
$q \in [3]$;
so ${\sf \Lambda}(\frac{\textstyle 1}
                       {\textstyle 3},
                  {\sf f}_{1}))
    =\frac{\textstyle 1}
          {\textstyle 3}$,
${\sf \Lambda}(\frac{\textstyle 1}
                    {\textstyle 3},
               {\sf f}_{2})
    =\frac{\textstyle 2}
          {\textstyle 3}$
and          
${\sf \Lambda}(\frac{\textstyle 1}
                    {\textstyle 3},
           {\sf f}_{3})
    = 1$.       
The payment function ${\sf P}$
gives payment $1$
to the player, if any, choosing the strictly highest quality, 
or gives payment $\frac{\textstyle 1}
                       {\textstyle 2}$
to each player in case of a tie;
so
${\sf P}_{i}(1,1) 
 = 
 {\sf P}_{i}(2,2) 
 = 
 {\sf P}_{i}(3,3)
 = 
 \frac{\textstyle 1}
      {\textstyle 2}$
for each player $i \in [2]$,      
${\sf P}_{1}(2,1) 
 = 
 {\sf P}_{2}(1,2)
 =
 {\sf P}_{1}(3,1)
 =
 {\sf P}_{2}(1,3)
 =
 {\sf P}_{1}(3, 1)
 =
 {\sf P}_{2}(1,3)
 = 
 1$
and
${\sf P}_{2}(2,1) 
 = 
 {\sf P}_{1}(1,2)
 =
 {\sf P}_{2}(3,1)
 =
 {\sf P}_{1}(1,3)
 =
 {\sf P}_{2}(3, 1)
 =
 {\sf P}_{1}(1,3)
 = 
 0$.
 Note that these payment functions are not oblivious
 as the payment to a player choosing a particular quality
 depends on the numbers of players
 choosing higher qualities. 
 We check that the game
 neither has the {\it FIP}
 nor a pure Nash equilibrium:
 \begin{itemize}

\item
If player 1 chooses 1, 
then player 2 gets utility 
$\frac{\textstyle 1}
      {\textstyle 2} - \frac{\textstyle 1}
                            {\textstyle 3}
 = 
 \frac{\textstyle 1}
      {\textstyle 6}$ 
when choosing 1,
$1 - \frac{\textstyle 2}
          {\textstyle 3}
= \frac{\textstyle 1}
       {\textstyle 3}$
when choosing 2,
and $1 - 1 = 0$
when choosing 3. 
So
player 2 chooses 2.

\item
If player 1 chooses 2, 
then player 2 gets utility 
$0 - \frac{\textstyle 1}
          {\textstyle 3}
 = - \frac{\textstyle 1}
          {\textstyle 3}$
when choosing 1, 
$\frac{\textstyle 1}
      {\textstyle 2}
 -
 \frac{\textstyle 2}     
      {\textstyle 3}
 =
 - \frac{\textstyle 1}
        {\textstyle 6}$
when choosing 2, 
and $1 - 1 = 0$ 
when choosing 3. 
So
player 2 chooses 3.

\item
If player 1 chooses 3, 
then player 2 gets utility 
$0 - \frac{\textstyle 1}
          {\textstyle 3}
 = - \frac{\textstyle 1}
          {\textstyle 3}$
when choosing 1,
$0 - \frac{\textstyle 2}
          {\textstyle 3}
 = 
 - \frac{\textstyle 2}
       {\textstyle 3}$
when choosing 2, 
and 
$\frac{\textstyle 1}
      {\textstyle 2}
 - 1 = 
 - \frac{\textstyle 1}
        {\textstyle 2}$
when choosing 3. 
So player 2 chooses 1.

\end{itemize}
Since players are anonymous
and payments are player-invariant,
player 1 best-responds to player 2 
in an identical way. 
Now note that the best-responses form the cycle
$\langle 1, 2 \rangle \leadsto
 \langle 3, 2 \rangle \leadsto
 \langle 3, 1 \rangle \leadsto
 \langle 2, 1 \rangle \leadsto
 \langle 2, 3 \rangle \leadsto
 \langle 1, 3 \rangle \leadsto
 \langle 1, 2 \rangle$,
while quality vectors outside the cycle
are not pure Nash equilibria. 
Hence, there is no pure Nash equilibrium.
\end{proof}

\noindent
We continue to prove:

\begin{proposition}
\label{icalp counterexample 2}
There is a contest game
with mandatory participation,
proportional allocation
and arbitrary players
that has neither the {\it FIP}
nor
a pure Nash equilibrium.
\end{proposition}

\begin{proof}
Fix an integer parameter $k \geq 2$. 
Consider the contest game with players 1 and 2
qualities $1, 2, \ldots, Q$
with ${\sf f}_{q} = q$ for each $q \in [Q]$,
where $Q = k + 1$,
and $s_{1} = \frac{\textstyle 1}
                  {\textstyle 4k - 2 + \frac{\textstyle 1}
                                            {\textstyle k+1}}$
and $s_{2} = \frac{\textstyle 1}
                  {\textstyle 4k + 2 + \frac{\textstyle 1}
                                            {\textstyle k+1}}$.
Consider a quality vector $(q_{1}, q_{2})$.
Then,
{\small
\begin{eqnarray*}
      {\sf U}_{1}(q_{1}, q_{2})
& = & \frac{\textstyle q_{1}}
           {\textstyle q_{1} + q_{2}}
      -
      \frac{\textstyle 1}
           {\textstyle 4k-2 + \frac{\textstyle 1}
                                   {\textstyle k+1}}
      q_{1}
\end{eqnarray*}
}
and
{\small
\begin{eqnarray*}
      {\sf U}_{2}(q_{1}, q_{2})
& = & \frac{\textstyle q_{2}}
           {\textstyle q_{1} + q_{2}}
      -
      \frac{\textstyle 1}
           {\textstyle 4k + 2 + \frac{\textstyle 1}
                                   {\textstyle k+1}}
      q_{2}\, .       
\end{eqnarray*}
}
We check that
a best-response cycle is possible.
Consider a unilateral deviation of player $1$ to quality $q_{1}' > q_{1}$.
Then,
{
\small
\begin{eqnarray*}
      {\sf U}_{1}(q_{1}', q_{2})
      -
      {\sf U}_{1}(q_{1}, q_{2})
& = & \frac{\textstyle q_{1}'}
           {\textstyle q_{1}' + q_{2}}
      -
      \frac{\textstyle q_{1}}
           {\textstyle q_{1} + q_{2}}
     -
     (q_{1}' - q_{1})\,
     \frac{\textstyle 1}
           {\textstyle 4k-2 + \frac{\textstyle 1}
                                   {\textstyle k+1}} \\
& = & \frac{\textstyle (q_{1}' - q_{1})q_{2}}
           {\textstyle (q_{1}' + q_{2}) (q_{1} + q_{2})}
      -
      (q_{1}' - q_{1})\,
      \frac{\textstyle 1}
           {\textstyle 4k-2 + \frac{\textstyle 1}
                                   {\textstyle k+1}} \\
& = & (q_{1}' - q_{1})
      \left( \frac{\textstyle q_{2}}
                  {\textstyle (q_{1}' + q_{2}) (q_{1} + q_{2})}
             -
             \frac{\textstyle 1}
                  {\textstyle 4k-2 + \frac{\textstyle 1}
                                          {\textstyle k+1}}
      \right)\, .                                                                                                                           
\end{eqnarray*}     
}
Similarly,
for a unilateral deviation of player 2
to quality $q_{2}'$,
{
\small
\begin{eqnarray*}
      {\sf U}_{2}(q_{1}, q_{2}')
      -
      {\sf U}_{1}(q_{1}, q_{2})
& = & (q_{2}' - q_{2})
      \left( \frac{\textstyle q_{1}}
                  {\textstyle (q_{1} + q_{2}') (q_{1} + q_{2})}
             -
             \frac{\textstyle 1}
                  {\textstyle 4k+2 + \frac{\textstyle 1}
                                          {\textstyle k+1}}
      \right)\, . 
\end{eqnarray*}     
}
Consider the sequence of deviations
$(1, 1) \leadsto
 (1, 2) \leadsto
 (2, 2) \leadsto
 \ldots
 \leadsto
 (k-1, k) \leadsto
 (k, k) \leadsto
 (k, k+1) 
 $,
where players $2$ and $1$
alternate in taking steps.
We prove that these steps
are improvements:
\begin{itemize}

\item
Consider first the step
$(\kappa, \kappa) \leadsto
 (\kappa, \kappa + 1)$,
taken by player $2$,
where $1 \leq \kappa \leq k$.
Then,
{
\small
\begin{eqnarray*}
      {\sf U}_{2}(\kappa, \kappa + 1)
      -
      {\sf U}_{2}(\kappa, \kappa)
& = & 	\frac{\textstyle \kappa}
            {\textstyle (\kappa + (\kappa + 1))(\kappa + \kappa)}
      -
      \frac{\textstyle 1}
           {\textstyle 4k+2 + \frac{\textstyle 1}
                                   {\textstyle k+1}} \\
& = & \frac{\textstyle 1}
           {\textstyle 2(2\kappa + 1)}
      -
      \frac{\textstyle 1}
           {\textstyle 2 (2k+1) + \frac{\textstyle 1}
                                       {\textstyle k+1}} \\
& \geq & \frac{\textstyle 1}
           {\textstyle 2(2k + 1)}
      -
      \frac{\textstyle 1}
           {\textstyle 2 (2k+1) + \frac{\textstyle 1}
                                       {\textstyle k+1}} \\
& > & 0\, .          
\end{eqnarray*}
} 
So the step 
$(\kappa, \kappa) \leadsto
 (\kappa, \kappa + 1)$ is an improvement for player 2.

\item
Consider now the step
$(\kappa-1, \kappa) \leadsto
 (\kappa, \kappa)$,
taken by player $1$,
where $1 \leq \kappa \leq k$.
Then,
{
\small
\begin{eqnarray*}
      {\sf U}_{1}(\kappa, \kappa)
      -
      {\sf U}_{1}(\kappa - 1, \kappa)
& = & 	\frac{\textstyle \kappa}
            {\textstyle (\kappa + \kappa)((\kappa-1 + \kappa)}
      -
      \frac{\textstyle 1}
           {\textstyle 2 (2k-1) + \frac{\textstyle 1}
                                       {\textstyle k+1}} \\
& = & \frac{\textstyle 1}
           {\textstyle 2(2\kappa - 1)}
      -
      \frac{\textstyle 1}
           {\textstyle 2 (2k-1) + \frac{\textstyle 1}
                                       {\textstyle k+1}} \\
& \geq & \frac{\textstyle 1}
           {\textstyle 2(2k - 1)}
      -
      \frac{\textstyle 1}
           {\textstyle 2 (2k-1) + \frac{\textstyle 1}
                                       {\textstyle k+1}} \\
& > & 0\, .          
\end{eqnarray*}
} 
So the step 
$(\kappa - 1, \kappa) \leadsto
 (\kappa, \kappa + 1)$ is an improvement for player 1.

\end{itemize} 
So a unilateral deviation
to the immediately higher quality by a player is an improvement.
We can similarly prove that a unilateral deviation
to a higher quality by either player is an improvement.
In particular, 
no quality vector $(q_{1}, q_{2})$
with $q_{1} \leq k$ and $q_{2} \leq k+1$
is a pure Nash equilibrium.
We will prove that
there is an improvement cycle
starting with the quality vector $(k, k+1)$.
\begin{itemize}

\item
Consider first the unilateral deviation 
$(k, k+1) \leadsto (k-1, k+1)$
by player $1$
to quality $k-1$.
Then,
{  
\begin{eqnarray*}
&   & {\sf U}_{1}(k-1, k+1)
      -
      {\sf U}_{1}(k, k+1) \\
& = & - 
      \left(
      \frac{\textstyle k+1}
           {\textstyle ((k-1) + (k+1)) (k + k+1)}
      -
     \frac{\textstyle 1}
           {\textstyle 2(2k-1) + \frac{\textstyle 1}
                                   {\textstyle k+1}}
     \right) \\
& = & - \frac{\textstyle k+1}
           {\textstyle 2k (2k+1)}
      +
      \frac{\textstyle 1}
           {\textstyle 2(2k-1) + \frac{\textstyle 1}
                                      {\textstyle k+1}}\, .                                                       
\end{eqnarray*}     
}
Thus,
${\sf U}_{1}(k-1, k+1) > {\sf U}_{1}(k, k+1) > 0$
if and only if
{
\small
\begin{eqnarray*}
      (k+1) \left[ 2 (2k-1) + \frac{\textstyle 1}
                              {\textstyle k+1}
            \right]
& < & 2k (2k+1)
\end{eqnarray*}
}
or
{
\small
\begin{eqnarray*}
      2 (k+1) (2k-1) + 1
& < & 2k (2k+1)
\end{eqnarray*}
}
which is verified directly.
Hence,
the unilateral deviation
$(k, k+1) \leadsto (k-1, k+1)$
by player $1$
is an improvement.

\item
Consider now the unilateral deviation 
$(k-1, k+1) \leadsto (k-1, k)$
by player $2$
to quality $k$.
Then,
{  
\begin{eqnarray*}
&   & {\sf U}_{2}(k-1, k)
      -
      {\sf U}_{2}(k-1, k+1) \\
& = & - 
      \left(
      \frac{\textstyle k-1}
           {\textstyle ((k-1) + (k+1)) (k + (k-1))}
      -
     \frac{\textstyle 1}
           {\textstyle 2(2k+1) + \frac{\textstyle 1}
                                   {\textstyle k+1}}
     \right) \\
& = & - \frac{\textstyle k-1}
           {\textstyle 2k (2k-1)}
      +
      \frac{\textstyle 1}
           {\textstyle 2(2k+1) + \frac{\textstyle 1}
                                      {\textstyle k+1}}\, .                                                       
\end{eqnarray*}     
}
Thus,
${\sf U}_{2}(k-1, k) > {\sf U}_{2}(k-1, k+1) > 0$
if and only if
{
\small
\begin{eqnarray*}
      (k-1) \left[ 2 (2k+1) + \frac{\textstyle 1}
                                   {\textstyle k+1}
            \right]
& < & 2k (2k-1)
\end{eqnarray*}
}
or
{
\small
\begin{eqnarray*}
      2 (k-1) (2k+1) + \frac{\textstyle k-1}
                            {\textstyle k+1}
& < & 2k (2k-1)
\end{eqnarray*}
}
which is verified directly.
Hence,
the unilateral deviation
$(k-1, k+1) \leadsto (k-1, k)$
by player $2$
is an improvement.

\item
Now the unilateral deviation $(k-1, k) \leadsto (k, k)$ by player 1
is an improvement
as it is a deviation from a lower quality
to a higher.
The unilateral deviation $(k, k) \leadsto (k, k+1)$ by player 1
is an improvement for the same reason.
Thus,
we get the improvement cycle
$(k, k+1) \leadsto
 (k-1, k+1) \leadsto
 (k-1, k) \leadsto
 (k, k) \leadsto
 (k, k+1)$.

\end{itemize}
Finally, 
note that $(k+1, k+1)$
is not a pure Nash equilibrium
since the unilateral deviation of player $1$ to strategy $k$
is an improvement:
\begin{eqnarray*}
      {\sf U}_{1}(k, k+1) -
      {\sf U}_{1}(k+1, k+1)
& = & -
      \left( \frac{\textstyle k+1}
                  {\textstyle (k + k+1) 2\, (k+1)}
             -
             \frac{\textstyle 1}
                  {\textstyle 2(2k-1) + \frac{\textstyle 1}
                                             {\textstyle k+1}} 
      \right) \\
& = & -
      \left( \frac{\textstyle 1}
                  {\textstyle 2\, (2k+1)}
             -
             \frac{\textstyle 1}
                  {\textstyle 2(2k-1) + \frac{\textstyle 1}
                                             {\textstyle k+1}} 
      \right) \\
& = & \frac{\textstyle 1}
           {\textstyle 2(2k-1) + \frac{\textstyle 1}
                                             {\textstyle k+1}}
      -
      \frac{\textstyle 1}
                  {\textstyle 2\, (2k+1)} \\
& > & 0\, ,                                                                                                                                                                                                                                                                                 
\end{eqnarray*}
since
$2(2k-1) + \frac{\textstyle 1}
                {\textstyle k+1}
 <
 2 (2k+1)$.
The claim follows.
\end{proof}

\begin{framed}
\begin{open problem}
Determine the precise complexity
of {\sc $\exists$PNE with Player-Invariant Payments}
and {\sc $\exists$PNE with Proportional Allocation and Arbitrary Players}.
We are tempted to conjecture that
both are ${\mathcal{NP}}$-complete.	
\end{open problem}	
\end{framed}

\noindent
We now turn to player-specific payments.
We show:

\begin{proposition}
\label{counter example}
There is a contest game
with player-specific payments 
and anonymous players
that has neither the {\it FIP}
nor a pure Nash equilibrium.
\end{proposition}

\begin{proof}
Consider
the contest game
with two players $1$ and $2$,
and two qualities $1$ and $2$
with ${\sf f}_{1} = 1$ and
${\sf f}_{2} = 2$.
Assume a skill-effort function
${\sf \Lambda}(1, {\sf f}_{q})
 = {\sf f}_{q}$
for all qualities $q \in [Q]$;
so
${\sf \Lambda}(1, {\sf f}_{1}) = 1$ and
${\sf \Lambda}(1, {\sf f}_{2}) = 2$. 
Similarly to {\it Matching Pennies,} 
player 1 has big payment when alone on a quality, 
else very small,
and player 2 has big payment when not alone,
else very small.
Formally,
define
${\sf P}_{1}(1, 1)
 =
 {\sf P}_{1}(2, 2)
 = 10^{3}$
${\sf P}_{1}(1,2)
 =
 {\sf P}_{1}(2,1)
 =
 10$,
${\sf P}_{2}(1,2)
 =
 {\sf P}_{2}(2,1)
 = 10^{3}$	
and
${\sf P}_{2}(1, 1)
 =
 {\sf P}_{2}(2,2)
 = 10$.
We check that there is no pure Nash equilibrium:
\begin{itemize}

\item
If player 1 chooses 1,
then player 2 gets utility
$10^3 - 1$ when choosing 1,
and $10 - 2 = 8$
when choosing 2.
So player 2 chooses 2.

\item
If player 1 chooses 2,
then player 2 gets utility 
$10^{3} - 1$ when choosing 1,
and $10 - 1 = 9$
when choosing 2.
So player 2 chooses 1.

\item
If player 2 chooses 1,
then player 1 gets utility
$10 - 1 = 9$
when choosing 1,
and $10^{3} - 2$
when choosing 2.
So player 1 chooses 2.

\item
If player 2 chooses 2,
then player 1 gets utility
$10^{3} - 1$
when choosing 1,
and $10 - 2 = 8$
when choosing 2.
So player 1 chooses 1.

\end{itemize}
Now note that the best-responses form the cycle
$\langle 1, 2 \rangle \leadsto
 \langle 1, 1 \rangle \leadsto
 \langle 2, 1 \rangle \leadsto
 \langle 2, 2 \rangle \leadsto
 \langle 1, 2 \rangle$,
while quality vectors outside the cycle
are not Nash equilibria. 
Hence, there is no pure Nash equilibrium.
\end{proof}

\noindent
We continue to show:

\begin{theorem}
\label{np completeness}
{\sc $\exists$PNE with Player-Specific Payments}
is ${\cal NP}$-complete,
even if players are anonymous.
\end{theorem}

\begin{proof}
{\sc $\exists$PNE with Player-Specific Payments}
$ \in {\cal NP}$
since one
can guess a quality vector
and verify the conditions
for a pure Nash equilibrium.
To prove ${\cal NP}$-hardness,
we reduce from the ${\cal NP}$-complete problem
of deciding the existence of a pure Nash equilibrium 
in a (finite) succinctly represented
strategic game~\cite[Theorem 2.4.1]{S74}.
So consider such a game
with $n$ players, $m$ strategies and
payoff functions
$\{ {\sf F}_{i} \}_{i \in [n]}$ 
represented by a poly-time algorithm
computing, 
for a pair
of a profile ${\bf s}$ and a player $i \in [n]$,
the payoff ${\sf F}(i, {\bf s})$ of player $i$ in ${\bf s}$.
Construct a contest game with
$n$ players, $Q=m$,
so that the quality vectors 
coincide with pure profiles
of the strategic game.
Define the payment function
as ${\sf P}_{i}(i, {\bf q})
    =
    {\sf F}_{i}(i, {\bf s}) + \Lambda (s_{i}, {\sf f}_{q})$ 
for a player $i$ and a strategy vector ${\bf q}$; 
thus,
${\sf U}_{i}({\bf q}) = {\sf F}_{i}(i, {\bf s})$.
${\cal NP}$-hardness follows.
\end{proof}

\section{
Proportional Allocation
}
\label{voluntary participation proportional allocation}

\subsection{
Anonymous Players
}

\noindent
We show:    

\begin{theorem}
\label{sto paderborn}
The contest game
with proportional allocation,
voluntary participation
and anonymous players
has the {\it FIP}
and two 
pure Nash equilibria.
\end{theorem}

\begin{proof}
It suffices to prove that there is no cycle
in the quality improvement graph.
Recall that voluntary participation means
${\sf f}_{1} = 0$.
We prove that improvement is possible
only if,
subject to an exception,
the deviating player is switching
from a higher quality to a lower quality:

\begin{lemma}[
No Switch from Lower Quality to Higher Quality
]
\label{kyriakatiko}
Fix a quality vector ${\bf q}$
and two distinct qualities
$\widetilde{q}, \widehat{q} \in [Q]$
with $\widetilde{q} < \widehat{q}$.
In an improvement step of a player out of ${\bf q}$,
${\sf N}_{{\bf q}}(\widetilde{q})$ increases and
and ${\sf N}_{{\bf q}}(\widehat{q})$ decreases.
\end{lemma}

\begin{proof}
\noindent
Denote ${\sf f}_{\widetilde{q}} = \beta$, 
${\sf f}_{\widehat{q}} = \gamma > \beta$,
$\chi =
 \sum_{q \in [Q] \setminus \{ \widetilde{q}, \widehat{q} \} }
  {\sf N}_{{\bf q}}(q) \geq 0$
and
${\sf A} =
 \sum_{q \in [Q] \setminus \{ \widetilde{q}, \widehat{q} \} }
  {\sf N}_{{\bf q}}(q)\,
  {\sf f}_{q} \geq 0$.
Denote the loads
on qualities $\widetilde{q}$ and $\widehat{q}$
as $x$ and $y$, respectively;
thus,
$y = n - \chi - x$
We shall abuse notation to denote 
the quality vector ${\bf q}$
as $(x, y)$.
\begin{itemize}

\item[{\sf (D1)}]
A deviation of a player from $\widehat{q}$
to $\widetilde{q}$
will be depicted as
\textcolor{red}{\Shortstack{$(x, y)$ \\ 
                             \rotatebox{315}{$\leadsto$} \\ 
                            $(x+1, y-1)$}}
with $x \geq 0$ and $y \geq 1$,
so as to guarantee the existence of at least one player
$i \in {\sf Players}_{{\bf q}}(\widehat{q})$.
Call such a deviation
{\it rightward\&downward}.

\item[{\sf (D2)}]
A deviation of a player
from $\widetilde{q}$ to $\widehat{q}$
will be depicted as
\textcolor{red}{\Shortstack{$(x-1, y+1)$ \\ 
                             \rotatebox{135}{$\leadsto$} \\ 
                            $(x, y)$}}
with $y \geq 0$ and $x \geq 1$,
so as to guarantee the existence of at least one player
$i \in {\sf Players}_{{\bf q}}(\widetilde{q})$.
Call such a deviation
{\it leftward\&upward}.

\end{itemize}
\noindent
Note that a rightward\&downward deviation
is an improvement for the deviating player
if and only if
the reverse leftward\&upward improvement step
is not an improvement for her.
We shall prove that a rightward\&downward deviation
\textcolor{red}{\Shortstack{$(x, y)$ \\ 
                             \rotatebox{315}{$\leadsto$} \\ 
                            $(x+1, y-1)$,}}
with $x \geq 0$ and $y \geq 1$,
is an improvement
unless $(x, y) = (n-1, 1)$.
Consider a player $i \in {\sf Players}_{(x,y)}(\widehat{q})$.
We proceed by case analysis.

\begin{enumerate}

\item
Assume first that $\widetilde{q} \neq 1$,
so that ${\sf f}_{\widetilde{q}} > 0$,
implying that ${\sf f}_{\widehat{q}} > 0$ as well.
So, in this case,
denominators in proportional allocation fractions 
are always strictly positive;
as we shall see in the analysis
for the case
$\widetilde{q}$, 
this is a crucial property. 
We have that
{
\small
\begin{eqnarray*}
      {\sf U}_{i}((x,y))
& = & \frac{\textstyle \gamma}
           {\textstyle {\sf A} +
                       x \beta +
                       (n- \chi - x) \gamma} - \gamma 
\end{eqnarray*}
}
and
{
\small
\begin{eqnarray*}
      {\sf U}_{i}((x+1,y-1))
& = & \frac{\textstyle \beta}
           {\textstyle {\sf A} +
                       (x+1) \beta +
                       (n - \chi - x - 1) \cdot \gamma} - \beta\, .
\end{eqnarray*}
}
\textcolor{red}{\Shortstack{$(x, y)$ \\ 
                             \rotatebox{315}{$\leadsto$} \\ 
                            $(x+1, y-1)$}}
is an improvement when
${\sf U}_{i}((x+1, y-1))
 >
 {\sf U}_{i}((x,y))$, or
\begin{eqnarray*}
      \frac{\textstyle \beta}
           {\textstyle {\sf A} + (x+1) \beta + (n- \chi-x-1) \gamma} - \beta
& > & \frac{\textstyle \gamma}
           {\textstyle {\sf A} + x \beta + (n- \chi - x) \gamma} - \gamma\, ,
\end{eqnarray*} 
or
\begin{eqnarray*}
      - \beta\,
      \frac{\textstyle {\sf A} + x \beta + (n-\chi -x-1) \gamma}
             {\textstyle {\sf A} + (x + 1) \beta + (n-\chi -x-1) \gamma}
& > & - \gamma\,
      \frac{\textstyle {\sf A} + (x-1) \beta + (n-\chi-x) \gamma}
           {\textstyle {\sf A} + x \beta + (n-\chi-x) \gamma}\, .
\end{eqnarray*}
Since both denominators are strictly positive
for every quality vector $(x,y)$,
the last is equivalent to
{
\small
\begin{eqnarray*}
&   & \beta\,
      [ {\sf A} + x \beta + (n-\chi-x) \gamma - \gamma]
      [ {\sf A} + x \beta + (n- \chi - x) \gamma ] \\
& < & \gamma\,
      [ {\sf A} + (x-1) \beta + (n-x-y) \gamma ]
      [ {\sf A} + (x+1) \beta + (n-\chi -x-1) \gamma ]  	
\end{eqnarray*}
}
or
{
\small
\begin{eqnarray*}
&   & \beta\,
      [ {\sf A} + x \beta + (n-\chi -x) \gamma ]^{2}
      -
      \beta \gamma
      [ {\sf A} + x \beta + (n-\chi-x) \gamma ]      \\
& < & 	\gamma\,
       [ {\sf A} + x \beta + (n-\chi - x) \gamma - 1]
       [ {\sf A} + x \beta + (n-\chi -x) \gamma + \beta - \gamma ] \\
& = & \gamma
      \left( [ {\sf A} + x \beta + (n-\chi - x) \gamma ]^{2}      
             -
             [ {\sf A} + x \beta + (n- \chi - x) \gamma ]
             +
             (\beta - \gamma)
             [ {\sf A} + x \beta + (n- \chi - x) \gamma ]
             -
             (\beta - \gamma) 
      \right)                            \\
& = & \gamma
      \left( [ {\sf A} + x \beta + (n- \chi - x) \gamma ]^{2}
             -
             \gamma
             [ {\sf A} + x \beta + (n- \chi -x) \gamma ]
             +
             (\gamma - \beta)
      \right)                           \\
& = & \gamma
      [ {\sf A} + x \beta + (n- \chi -x) \gamma ]^{2}
      -
      \gamma^{2}
      [ {\sf A} + x \beta + (n- \chi -x) \gamma ]
      +
      \gamma (\gamma - \beta) 
\end{eqnarray*}
}
or
{
\small
\begin{eqnarray*}
      (\gamma - \beta)
      [{\sf A} + x \beta + (n-\chi - x) \gamma]^{2}
      -
      \gamma (\gamma  - \beta )
      [{\sf A} + x \beta + (n-\chi -x) \gamma ]
      +
      \gamma (\gamma - \beta)
& > & 0\, .   	
\end{eqnarray*}
}
Since $\gamma > \beta$,
the last inequality is equivalent to
{
\small
\begin{eqnarray*}
      [{\sf A} + x \beta + (n- \chi -x) \gamma ]^{2}
      -
      \gamma
      [{\sf A} + x \beta + (n- \chi - x) \gamma ]
      +
      \gamma
& > & 0\,    	
\end{eqnarray*}
}
or 
{
\small
\begin{eqnarray*}
      [{\sf A} + x \beta + (n-\chi - x) \gamma ]^{2}
& > & \gamma
      \left( [{\sf A} + x \beta + (n- \chi - x) \gamma ]
             -
             \beta
      \right)\, .
\end{eqnarray*}
}
Since $n- \chi - x \geq 1$,
it follows that
${\sf A} + x \beta + (n- \chi - x) \gamma \geq \gamma$,
which implies
\begin{eqnarray*}
[{\sf A} + x \beta +(n- \chi - x) \gamma]^2 
& \geq & \gamma [{\sf A} + x \beta + (n- \chi- x) \gamma ] \\
& > &    \gamma [{\sf A} + x \beta + (n-\chi - x) \gamma - \beta ]\, ,
\end{eqnarray*}                     
since $\beta > 0$.
It follows that
\textcolor{red}{\Shortstack{$(x, y)$ \\ 
                             \rotatebox{315}{$\leadsto$} \\ 
                            $(x+1, y-1)$}}
is an improvement,
implying that
\textcolor{red}{\Shortstack{$(x-1, y+1)$ \\ 
                             \rotatebox{135}{$\leadsto$} \\ 
                            $(x, y)$,}}
with $x \geq 0$ and $y > 0$,
is not.

\item
Assume now that $\widetilde{q} = 1$,
so that ${\sf f}_{\widetilde{q}} = 0$.
Then,
it is no longer the case
that denominators in proportional allocation fractions
are always strictly positive.
Specifically,
when $x = n-1$ and $y = 1$,
some denominator becomes $0$
as we shall see.
So the case $x=n-1$ and $y=1$
will require special handling.
We proceed with the details.
In all cases, we have that
{
\small
\begin{eqnarray*}
      {\sf U}_{i}((x,y))
& = & \frac{\textstyle \beta}
           {\textstyle {\sf A} + x \cdot 0 +
                       +
                       (n- \chi - x) \cdot \gamma} - \beta\ \
      =\ \
      \beta\,
      \left( \frac{\textstyle 1}
           {\textstyle {\sf A} + (n-\chi - x) \gamma} - 1
      \right)\, ,                 
\end{eqnarray*}
}
and
{
\small
\begin{eqnarray*}
      {\sf U}_{i}((x+1,y-1))
& = & \frac{\textstyle 0}
           {\textstyle {\sf A} + (x+1) \cdot 0 +
                       (n-\chi - x-1) \cdot \gamma} - 0\ \
      =\ \
      \frac{\textstyle 0}
           {\textstyle {\sf A} + (n-\chi -x - 1) \gamma}                        
\end{eqnarray*}
}
Note that if $y = 1$ and $x=n-1$,
then in ${\sf U}_{i}(x+1, y-1)$,
${\sf A} = 0$, $\chi = 0$ and $n-\chi - x - 1 =0$,
so that the denominator in the fraction of
${\sf U}_{i}(x+1, y-1)$ becomes also $0$,
making the fraction indeterminate;
in this case,
${\sf U}_{i}((x+1,y-1))$ is $0$ by the way 
indeterminacy is removed.
In all other cases,
the denominator is strictly positive,
which results again in
${\sf U}_{i}((x+1, y-1)) = 0$.
So, 
${\sf U}_{i}((x+1, y-1)) = 0$ in every case.
\textcolor{red}{\Shortstack{$(x, y)$ \\ 
                             \rotatebox{315}{$\leadsto$} \\ 
                            $(x+1, y-1)$}}
is an improvement
when
${\sf U}_{i}((x+1, y-1)) > {\sf U}_{i}((x, y))$
or
\begin{eqnarray*}
       \frac{\textstyle 1}
            {\textstyle {\sf A} +
                        (n- \chi - x) \gamma}      
& < & 1\, .
\end{eqnarray*}
\begin{itemize}

\item
\underline{$(x, y) = (n-1, 1)$:}
Then,
the denominator in ${\sf U}_{i}((x, y)$ becomes $1$,
resulting to ${\sf U}_{i}((x, y))$ is also $0$,
implying that neither the rightward\&downward deviation
\textcolor{red}{\Shortstack{$(x, y)$ \\ 
                             \rotatebox{315}{$\leadsto$} \\ 
                            $(x+1, y-1)$}}
nor the leftward\&upward deviation
\textcolor{red}{\Shortstack{$(x-1, y+1)$ \\ 
                             \rotatebox{135}{$\leadsto$} \\ 
                            $(x, y)$}}
is an improvement.

\item
\underline{$(x, y) \neq (n-1, 1)$:}
Thus, either $x = n$ or $x \leq n-2$.
We proceed by case analysis.

\begin{itemize}

\item
\underline{$x=n$:}
Then,
$y=0$ and
there can be no 
\textcolor{red}{\Shortstack{$(x, y)$ \\ 
                             \rotatebox{315}{$\leadsto$} \\ 
                            $(x+1, y-1)$}}
deviation out of $(n, 0)$. 

\item
\underline{$x \leq n-2$:}
Then,
$n - \chi - x \geq 2$
and
${\sf A} + (n- \chi - x) \gamma \geq 2 \gamma > 2$.
It follows that
the necessary and sufficient condition for an improvement holds.

\noindent
It follows that, unless $(x, y) = (n-1, 1)$,
the rightward\&downward deviation
\textcolor{red}{\Shortstack{$(x, y)$ \\ 
                             \rotatebox{315}{$\leadsto$} \\ 
                            $(x+1, y-1)$}}
is an improvement,
implying that
the leftward\&upward deviation
\textcolor{red}{\Shortstack{$(x-1, y+1)$ \\ 
                             \rotatebox{135}{$\leadsto$} \\ 
                            $(x, y)$}}
is not.

\end{itemize}

\end{itemize}
Hence,
rightward\&downward deviations are improvements
except when $(x,y) = (n-1, 1)$.

\end{enumerate}

\end{proof}

\noindent
It follows
that the quality improvement graph
has two sinks,
representing two pure Nash equilibria:
\begin{itemize}

\item
The node $(n-1, 1)$,
corresponding to 
${\sf N}_{{\bf q}}(1) = n-1$,
${\sf N}_{{\bf q}}(2) = 1$
and ${\sf N}_{{\bf q}}(q) = 0$
for each quality $q \in [Q]$ with $q > 2$.

\item
The node $(n, 0)$,
corresponding to 
${\sf N}_{{\bf q}}(1) = n$
and ${\sf N}_{{\bf q}}(q) = 0$
for each quality $q \in [Q]$ with $q > 1$.
This node is unreachable by improvement steps.

\end{itemize}

\remove{\textcolor{magenta}{Note that there is
{\it (i)}
a rightward horizontal improvement step
out of nodes $(x, y)$ 
with $0 \leq x \leq n-1$ and $x+y < n$,
{\it (ii)}
an upward vertical improvement step
out of nodes $(x,y)$ with 
$0 \leq x \leq n-1$ and $x+y < n$,
{\it (iii)}
a rightward\&downward  improvement step
out of nodes $(x,y)$ with
$x > 0$ and $x+y \leq n$,
and
there is no improvement step out of $(n,0)$ and $(n-1, 1)$.
It follows that
nodes $(n,0)$ and $(n-1, 1)$ are the only sinks,
corresponding to two pure Nash equilibria.}
\textcolor{blue}{{\bf CHECK VERY CAREFULLY.}}}
\end{proof}

\noindent
Under mandatory participation,
it no longer holds that
${\sf f}_{1} = 0$,
and Case {\bf 2.} in the proof
of Lemma~\ref{kyriakatiko} does not arise;
as a result,
the node $(n-1, 1)$,
corresponding to
${\sf N}_{{\bf q}}(1) = n-1$,
${\sf N}_{{\bf q}}(2) = 1$
and ${\sf N}_{{\bf q}}(q) = 0$
for each quality $q \in [Q]$ with $q > 2$,
is not a sink anymore
since the unilateral deviation of a player from quality 2
to quality 1 is now an improvement
since ${\sf f}_{1} > 0$.
So we have now a unique pure Nash equilibrium,
where all players choose quality $1$.
The rest of the proof of Theorem~\ref{sto paderborn}
transfers over.  
Hence,
we have:

\begin{theorem}
\label{wow paderborn}
The contest game
with proportional allocation,
mandatory participation
and anonymous players
has the {\it FIP}
and a unique pure Nash equilibrium.
\end{theorem}

\noindent
Given the counter-example contest game
in Proposition~\ref{icalp counterexample 2},
Theorem~\ref{wow paderborn}
establishes a {\it separation} with respect to the {\it FIP} property
and the existence of a pure Nash equilibrium
between arbitrary players and anonymous players,
under mandatory participation and proportional allocation.
Theorems~\ref{sto paderborn} and~\ref{wow paderborn}
imply:

\remove{
\noindent
\textcolor{blue}{
A {\it generalized ordinal potential}~\cite{MS96}
for a finite game
is a function $\Phi$,
mapping profiles
(that is,
quality vectors for the special case of the contest game) 
to numbers, such that
for every player $i \in [n]$,
for every pair of strategies $q_i$ and $q'_{i}$
and for every partial profile ${\bf q}_{-i}$,
${\sf U}_{i}(q_{i}, {\bf q}_{-i}) 
 > 
 {\sf U}_{i}(q_{i}', {\bf q}_{-i})$
implies
$\Phi (q_{i}, {\bf q}_{-i})
 >
 \Phi (q_{i}', {\bf q}_{-i})$.
So a potential
is a strengthening of a generalized ordinal potential;
nevertheless,  
every game with a generalized ordinal potential
still has a pure Nash equilibrium. 
By~\cite[Lemma 2.5]{MS96},
a finite game has a generalized ordinal potential
if and only if
it has the {\it FIP}.
Hence,}}

\begin{corollary}
\label{sto paderborn corollary}
The contest game
with
proportional allocation
and
anonymous players has 
a generalized ordinal potential.
\end{corollary}

\subsection{
Mandatory Participation
}

We show:

\begin{theorem}
\label{natasa}
There is a ${\sf \Theta} (1)$ algorithm
that solves 
{\sc $\exists$PNE with Proportional Allocation and Arbitrary Players}
with
lower-bounded skills 
$\min_{i \in [n]} 
   s_{i} 
 \geq  
 \frac{\textstyle {\sf f}_{2}}
      {\textstyle {\sf f}_{2} - {\sf f}_{1}}$
and
skill-effort functions
$\Lambda (s_{i}, {\sf f}_{q})
 = s_{i} {\sf f}_{q}$, 
for all players $i \in [n]$
and qualities $q \in [Q]$. 
\end{theorem}

\begin{proof} 
By definition of utility
and mandatory participation,
the utility of each player $i \in [n]$ is 
more than $- s_{i} {\sf f}_{1}$.
If player $i$ deviates to $2$,
its utility will be less than ${\sf f}_{2} - {\sf f}_{2} s_{i} 
= - {\sf f}_{2} (s_{i}- 1)$.
The assumption implies that
$- {\sf f}_{2} (s_{i} - 1) 
 \leq - {\sf f}_{1} s_{i}$
for all players $i \in [n]$.
So player $i$ does not want to switch
to quality $2$.
Since efforts are increasing,
for all qualities $q$ with $2 < q \leq Q$,
the utility of player $i$ when she deviates to $q$
will be less than
$ - {\sf f}_{q} (s_{i} - 1)
 < 
 - {\sf f}_{2} (s_{i} - 1)
 \leq
 - {\sf f}_{1} s_{i}$,
by the assumption.
So
player $i$ does not want to switch 
to any quality $q > 2$ either.
Hence,
assigning all players to quality $1$
is a pure Nash equilibrium.
\end{proof}

\noindent
Since 
$\frac{\textstyle {\sf f}_{2}}
      {\textstyle {\sf f}_{2} - {\sf f}_{1}} > 1$,
the assumption
made for Theorem~\ref{natasa}
that
all skills are lower-bounded
by $\frac{\textstyle {\sf f}_{2}}
      {\textstyle {\sf f}_{2} - {\sf f}_{1}}$
in Theorem~\ref{natasa}
cannot hold for
anonymous players
where $s_{i} = 1$
for all players $i \in [n]$.
This assumption is reasonable
for real contests for crowdsourcing reviews
where a minimum skill
is required for reviewers
in order to eliminate the risk
of receiving inferior solutions of low quality.
Indeed, crowdsourcing firms can target crowd contributors 
based on exhibiting skills, 
like performance in prior contests. 
Clearly,
the assumption made for Theorem~\ref{natasa},
enabling the existence of a pure Nash equilibrium,
could {\em not} hold
for the counter-example contest game
in Proposition~\ref{icalp counterexample 2}.

\remove{
\begin{theorem}
\label{natasa 1}
\textcolor{blue}{
There is a $\Theta (1)$ algorithm
that solves 
{\sc $\exists$PNE with Proportional Allocation and Anonymous Players}
with
skill-effort functions
$\Lambda (s_{i}, {\sc f}_{q}) = 
 {\sf f}_{q}$, 
for all players $i \in [n]$ and qualities
$q \in [Q]$, 
and under the assumption
${\sf f}_{2} - {\sf f}_{1} \geq 1$.
}
\end{theorem}

\begin{proof} 
\textcolor{blue}{
The algorithm assigns all players to quality $1$. 
By definition,
the utility of each player $i \in [n]$ is 
less than ${\sf f}_{1}$.
If player $i$ switches to quality $2$,
her utility will be smaller than $- {\sf f}_{2} + 1$.
The assumption implies that
player $i$ does not want to switch to quality $2$.
Since efforts are increasing,
for all qualities $q$ with $2 < q \leq Q$,
the utility of player $i$ when she switches to quality $q$
will be greater than
$- {\sf f}_{q} + 1
 <
 - {\sf f}_{2} + 1
 \leq
 - {\sf f}_{1}$,
by the assumption ${\sf f}_{2} - {\sf f}_{1} \geq 1$.
So
player $i$ does not want to switch to any quality
$q > 2$ either.
Hence,
assigning all players to quality $1$
is a pure Nash equilibrium.
}
\end{proof}

\noindent
\textcolor{blue}{
The assumption that
${\sf f}_{2} > {\sf f}_{1} + 1$
implies the more general assumption
that ${\sf f}_{q} - {\sf f}_{1} \geq 1$
for all qualities $q > 1$.
This is a suitable assumption for crowdsourcing contests
wishing to define and use qualities such that
the highest quality
is well-separated from the rest.}
}

\section{Three-Discrete-Concave Payments and Contiguity}
\label{kalia}

Say that the load vector
${\sf N}_{{\bf q}}$ 
is {\it contiguous}
if players $1$ to ${\sf N}_{{\bf q}}(1)$ 
choose quality $1$,
players 
${\sf N}_{{\bf q}}(1) + 1$ 
to ${\sf N}_{{\bf q}}(1) + {\sf N}_{{\bf q}}(2)$ choose quality $2$,
and so on till
players $\sum_{q \in [Q-1]} {\sf N}_{{\bf q}}(q) + 1$ to $n$ 
choose quality $q_{{\sf last}} \leq Q$
such that
for each quality $\widehat{q} > q_{{\sf last}}$,
${\sf N}_{{\bf q}}(\widehat{q}) = 0$;
so for any players $i$ and $k$, with $i < k$,
choosing distinct qualities
$q$ and $q'$, 
respectively,
we have $q < q'$.
Clearly,
a contiguous load vector determines by itself
which ${\sf N}_{{\bf q}}(q)$ players
choose each quality $q \in [Q]$.
\remove{
Given a contiguous load vector ${\sf N}_{{\bf q}}$,
denote, 
for each quality $q \in [Q]$ 
such that ${\sf Players}_{{\bf q}}(q) \neq \emptyset$,
the minimum and the maximum, respectively, 
player index $i \in {\sf Players}_{{\bf q}}(q)$
as ${\sf first}_{{\bf q}}(q)$ 
and ${\sf last}_{{\bf q}}(q)$,
respectively.
Clearly,
${\sf first}_{{\bf q}}(q) = 
 \sum_{\widehat{q} < q} {\sf N}_{{\bf q}}(\widehat{q}) + 1$
and
${\sf last}_{{\bf q}}(q) = 
 \sum_{\widehat{q} \leq q} {\sf N}_{{\bf x}}(\widehat{q})$;
so
${\sf first}_{{\bf q}}(1) = 1$ 
for ${\sf N}_{{\bf q}}(1) > 0$ 
and ${\sf last}_{{\bf q}}(Q) = n$
for ${\sf N}_{{\bf q}}(Q) > 0$.
}

Say that
an {\it inversion} occurs in a load vector ${\sf N}_{{\bf q}}$
if there are players $i$ and $k$ with $i < k$
choosing qualities $q_{i}$ and $q_{k}$, respectively,
with $q_{i} > q_{k}$;
thus, $s_{i} \geq s_{k}$
while ${\sf f}_{q_{i}} > {\sf f}_{q_{k}}$.
Call $i$ an {\it inversion witness};
call $i$ and $k$ an {\it inversion pair}.
Clearly,
no inversion occurs in a load vector ${\sf N}_{{\bf q}}$
if and only if
${\sf N}_{{\bf q}}$ is contiguous.

Given a contiguous load vector ${\sf N}_{{\bf q}}$,
denote, 
for each quality $q \in [Q]$ 
such that ${\sf Players}_{{\bf q}}(q) \neq \emptyset$,
the minimum and the maximum, respectively, 
player index $i \in {\sf Players}_{{\bf q}}(q)$
as ${\sf first}_{{\bf q}}(q)$ 
and ${\sf last}_{{\bf q}}(q)$,
respectively.
Clearly,
${\sf first}_{{\bf q}}(q) = 
 \sum_{\widehat{q} < q} {\sf N}_{{\bf q}}(\widehat{q}) + 1$
and
${\sf last}_{{\bf q}}(q) = 
 \sum_{\widehat{q} \leq q} {\sf N}_{{\bf x}}(\widehat{q})$;
so
${\sf first}_{{\bf q}}(1) = 1$ 
for ${\sf N}_{{\bf q}}(1) > 0$ 
and ${\sf last}_{{\bf q}}(Q) = n$
for ${\sf N}_{{\bf q}}(Q) > 0$.

Order the players so that
$s_{1} \geq s_{2} \geq \ldots \geq s_{n}$.
Recall that ${\sf f}_{1} < {\sf f}_{2} < \ldots < {\sf f}_{Q}$.
Represent a quality vector ${\bf q}$
as follows:
\begin{itemize}

\item
Use a {\it load vector}
${\sf N}_{{\bf q}}
 =
 \langle {\sf N}_{{\bf q}}(1), 
         {\sf N}_{{\bf q}}(2), 
         \ldots, 
         {\sf N}_{{\bf q}}(Q) 
 \rangle$.

\item
Specify
which ${\sf N}_{{\bf q}}(q)$ players
choose each quality $q \in [Q]$.

\end{itemize}
\noindent
To simplify notation,
we shall often omit to specify
the players choosing each quality $q \in [Q]$.
Thus,
we shall represent a quality vector ${\bf q}$
by the load vector
${\sf N}_{{\bf q}}$.

\subsection{Player-Specific Payments}

\noindent
Recall that
a player-specific payment function ${\sf P}_{i}({\bf q})$
can be represented by
a two-argument payment function ${\sf P}_{i}(i, {\bf q})$,
where $i \in [n]$ and ${\bf q}$ is a quality vector.
We start by defining:

\begin{definition}
\label{combinatorial condition 1}
A player-specific payment function
${\sf P}$
is {\it three-discrete-concave} if
for every player $i \in [n]$,
for every load vector ${\sf N}_{{\bf q}}$
and for every triple of qualities
$q_{i}$, $q_{k}$, $q \in [Q]$,
{
\small
\begin{eqnarray*}
& &
{\sf P}_{i}(i, ({\sf N}_{{\bf q}}(1), \ldots,
                   {\sf N}_{{\bf q}}(q_{i}), \ldots,
                   {\sf N}_{{\bf q}}(q_{k})-1, \ldots,
                   {\sf N}_{{\bf q}}'(q)+1, \ldots,
                   {\sf N}_{{\bf q}}'(Q)))
+ \\
& &
 {\sf P}_{i}(i, ({\sf N}_{{\bf q}}(1), 
             \ldots, 
             {\sf N}_{{\bf q}}(q_{i}) - 1, 
             \ldots, 
             {\sf N}_{{\bf q}}(q_{k}),
             \ldots,
             {\sf N}_{{\bf q}}(q) + 1, 
             \ldots, 
             {\sf N}_{{\bf x}}(Q)))	\\
& \leq & 
2\,
{\sf P}_{i}(i, ({\sf N}_{{\bf q}}(1), \ldots,
                        {\sf N}_{{\bf q}}(q_{i}), \ldots, 
                        {\sf N}_{{\bf q}}(q_{k}), \ldots,
                        {\sf N}_{{\bf q}}(q), \ldots,
                        {\sf N}_{{\bf q}}(Q)))\, .
\end{eqnarray*}
}
\end{definition}
\noindent
The inequality in Definition~\ref{combinatorial condition 1}
may be rewritten as
{
\small
\begin{eqnarray*}
& &
{\sf P}_{i}(i, ({\sf N}_{{\bf q}}(1), \ldots,
                   {\sf N}_{{\bf q}}(q_{i})-1, \ldots,
                   {\sf N}_{{\bf q}}(q_{k}), \ldots,
                   {\sf N}_{{\bf q}}(q)+1, \ldots,
                   {\sf N}_{{\bf q}}(Q)))
- \\
& &
{\sf P}_{i}(i, ({\sf N}_{{\bf q}}(1), \ldots,
                        {\sf N}_{{\bf q}}(q_{i}), \ldots, 
                        {\sf N}_{{\bf q}}(q_{k}), \ldots,
                        {\sf N}_{{\bf q}}(q), \ldots,
                        {\sf N}_{{\bf q}}(Q))) \\
& \leq &
{\sf P}_{i}(i, ({\sf N}_{{\bf q}}(1), \ldots,
                        {\sf N}_{{\bf q}}(q_{i}), \ldots, 
                        {\sf N}_{{\bf q}}(q_{k}), \ldots,
                        {\sf N}_{{\bf q}}(q), \ldots,
                        {\sf N}_{{\bf q}}(Q))) - \\
& &
{\sf P}_{i}(i, ({\sf N}_{{\bf q}}(1), 
             \ldots, 
             {\sf N}_{{\bf q}}(q_{i}), 
             \ldots, 
             {\sf N}_{{\bf q}}(q_{k})-1,
             \ldots,
             {\sf N}_{{\bf q}}(q) + 1, 
             \ldots, 
             {\sf N}_{{\bf q}}(Q)))\, .
\end{eqnarray*}
}

\remove{
A function 
$\Lambda: {\mathbb{R}}^{2}
          \rightarrow
          {\mathbb{R}}$
of two variables $y$ and $z$
{\it has increasing differences}
if for all $y, y'$ with $y \geq y'$,
the difference
$\Lambda (y, z) - \Lambda (y', z)$
is monotonically increasing in $z$;
that is,
for all $z \geq  z'$,
$\Lambda (y, z)
 -
 \Lambda (y', z)
 \geq
 \Lambda (y, z')
 -
 \Lambda (y', z')$.
\textcolor{blue}{A player-specific payment function ${\sf P}$
is permutation-invariant
if for every player $i \in [n]$,
for every permutation $\pi$ on $[Q]$
and for every load vector ${\sf N}_{{\bf q}}$
induced by a quality vector ${\bf q}$,
${\sf P}_{i}(i, {\sf N}_{{\bf q}}=
 {\sf P}_{i}(i, \pi ({\sf N}_{{\bf q}})$.}}
We show:

\begin{theorem}
\label{the most arbitrary}
There is a 
$\Theta \left( 
                n
                \cdot
                Q^{2}\,
                \binom{\textstyle n+Q-1}{\textstyle Q-1}
                \right)$
algorithm that solves 
{\sc $\exists$PNE with Player-Specific Payments}
for arbitrary players
and three-discrete-concave player-specific payments;
for constant $Q$, 
it is a $\Theta (n^{Q})$ 
polynomial algorithm.
\end{theorem}

\begin{proof}
\noindent
We start by proving:

\begin{proposition}[{\bf Contigufication Lemma for Player-Specific Payments}]
\label{transformation into contiguous}
For three-discrete-concave player-specific payments,  
any pair of
{\it (i)}
a pure Nash equilibrium 
${\sf N}_{{\bf q}}
 = \langle {\sf N}_{{\bf q}}(1),
           \ldots, 
           {\sf N}_{{\bf q}}(Q)
   \rangle$
and 
{\it (ii)}
player sets ${\sf Players}_{{\bf q}}(q)$
for each quality $q \in [Q]$,
can be transformed into a contiguous pure Nash equilibrium.
\end{proposition}

\begin{proof}
If no inversion occurs in ${\sf N}_{{\bf q}}$,
then ${\sf N}_{{\bf q}}$ is contiguous
and we are done.
Else take the earliest inversion witness $i$,
together with the earliest player $k$
such that $i$ and $k$ make an inversion.
We shall also consider a player
$\iota \in [n] \setminus \{ i, k \}$.
Since payments are player-specific,
{
\small
\begin{eqnarray*}
{\sf U}_{i}({\sf N}_{{\bf q}}) 
& = & 
 {\sf P}_{i}(i, {\sf N}_{{\bf q}})
 - 
 \Lambda (s_{i}, {\sf f}_{q_{i}})
\end{eqnarray*}
} 
and
{
\small
\begin{eqnarray*}
{\sf U}_{k}({\sf N}_{{\bf q}}) 
& = & {\sf P}_{k}(k, {\sf N}_{{\bf q}})
      - 
      \Lambda (s_{k}, {\sf f}_{q_{k}})\, .
\end{eqnarray*}
}
\begin{enumerate}

\item
Player $i$ does not want to switch to quality 
$q \neq q_{i}$
if and only if
{
\small
\begin{eqnarray*}
&      & {\sf P}_{i}(i, ({\sf N}_{{\bf q}}(1),
                \ldots
                {\sf N}_{{\bf q}}(q_{k}),
                \ldots,
                {\sf N}_{{\bf q}}(q_{i}),
                \ldots,
                {\sf N}_{{\bf q}}(Q)))
         - \Lambda (s_{i}, {\sf f}_{q_{i}}) \\
& \geq &
 {\sf P}_{i}(i, ({\sf N}_{{\bf q}}(1),
                \ldots
                {\sf N}_{{\bf q}}(q) + 1,
                \ldots,
                {\sf N}_{{\bf q}}(q_{i}) - 1,
             \ldots,
             {\sf N}_{{\bf q}}(Q)))
 - \Lambda (s_{i}, {\sf f}_{q})\, ,
\end{eqnarray*}
} 
or
{
\small
\begin{eqnarray*}
         \Lambda (s_{i}, {\sf f}_{q_{i}}) 
         - 
         \Lambda (s_{i}, {\sf f}_{q}) 
& \leq &
   {\sf P}_{i}(i, ({\sf N}_{{\bf q}}(1),
                \ldots
                {\sf N}_{{\bf q}}(q),
                \ldots,
                {\sf N}_{{\bf q}}(q_{i}),
                \ldots,
                {\sf N}_{{\bf q}}(Q))) - \\
&      &
 {\sf P}_{i}(i, ({\sf N}_{{\bf q}}(1), 
             \ldots, 
             {\sf N}_{{\bf q}}(q) + 1, 
             \ldots, 
             {\sf N}_{{\bf q}}(q_{i}) - 1, 
             \ldots, 
             {\sf N}_{{\bf q}}(Q)))
\hfill~~~~~
{\sf (C.1)}\, .
\end{eqnarray*}
}

\item
Player $k$ does not want to switch to quality $q \neq q_{k}$
if and only if
{
\small
\begin{eqnarray*}
&      & {\sf P}_{k}(k, ({\sf N}_{{\bf q}}(1),
                \ldots
                {\sf N}_{{\bf q}}(q_{k}),
                \ldots,
                {\sf N}_{{\bf q}}(q),
                \ldots,
                {\sf N}_{{\bf q}}(Q)))
 - \Lambda (s_{k}, {\sf f}_{q_{k}}) \\
& \geq &
 {\sf P}_{k}(k, ({\sf N}_{{\bf q}}(1),
                \ldots
                {\sf N}_{{\bf q}}(q_{k}) - 1,
                \ldots,
                {\sf N}_{{\bf q}}(q) + 1,
             \ldots,
             {\sf N}_{{\bf q}}(Q)))
 - \Lambda (s_{k}, {\sf f}_{q})\, ,
\end{eqnarray*}
} 
or
{
\small
\begin{eqnarray*}
         \Lambda (s_{k}, {\sf f}_{q_{k}}) 
         - 
         \Lambda (s_{k}, {\sf f}_{q})
& \leq &
   {\sf P}_{k}(k, ({\sf N}_{{\bf q}}(1),
                \ldots
                {\sf N}_{{\bf q}}(q_{k}),
                \ldots,
                {\sf N}_{{\bf q}}(q),
                \ldots,
                {\sf N}_{{\bf q}}(Q))) - \\
&      &
 {\sf P}_{k}(k, ({\sf N}_{{\bf q}}(1), 
             \ldots, 
             {\sf N}_{{\bf q}}(q_{k}) - 1, 
             \ldots, 
             {\sf N}_{{\bf q}}(q) + 1, 
             \ldots, 
             {\sf N}_{{\bf x}}(Q)))
\hfill~~~~~
{\sf (C.2)}\, .
\end{eqnarray*}
}

\item
Player $\iota$ does not want to switch to quality $q \neq q_{\iota}$
if and only if
{
\small
\begin{eqnarray*}
&      & {\sf P}_{\iota}(\iota, ({\sf N}_{{\bf q}}(1),
                \ldots
                {\sf N}_{{\bf q}}(q),
                \ldots,
                {\sf N}_{{\bf q}}(q_{\iota}),
                \ldots,
                {\sf N}_{{\bf q}}(Q)))
         - \Lambda (s_{\iota}, {\sf f}_{q_{\iota}}) \\
& \geq &
 {\sf P}_{\iota}(\iota, ({\sf N}_{{\bf q}}(1),
                \ldots
                {\sf N}_{{\bf q}}(q) + 1,
                \ldots,
                {\sf N}_{{\bf q}}(q_{\iota}) - 1,
             \ldots,
             {\sf N}_{{\bf q}}(Q)))
 - \Lambda (s_{\iota}, {\sf f}_{q})\, ,
\end{eqnarray*}
} 
or
{
\small
\begin{eqnarray*}
         \Lambda (s_{\iota}, {\sf f}_{q_{\iota}}) 
         - 
         \Lambda (s_{\iota}, {\sf f}_{q}) 
& \leq &
     {\sf P}_{\iota}(\iota, ({\sf N}_{{\bf q}}(1),
                \ldots
                {\sf N}_{{\bf q}}(q),
                \ldots,
                {\sf N}_{{\bf q}}(q_{\iota}),
                \ldots,
                {\sf N}_{{\bf q}}(Q))) - \\
&      &
  {\sf P}_{\iota}(\iota, ({\sf N}_{{\bf q}}(1), 
             \ldots, 
             {\sf N}_{{\bf q}}(q) + 1, 
             \ldots, 
             {\sf N}_{{\bf q}}(q_{\iota}) - 1, 
             \ldots, 
             {\sf N}_{{\bf q}}(Q)))
\hfill~~~~~~~
{\sf (C.3)}\, .
\end{eqnarray*}
}

\end{enumerate}

\noindent
\begin{framed}
\noindent
Swap the qualities chosen by players $i$ and
$k$;
so they now choose $q_{k}$
and $q_{i}$, respectively.
Choices of other players are preserved. 
\end{framed}
\noindent
Denote as ${\sf N}_{{\bf q}'}$ the resulting load vector;
clearly,
for each $\widehat{q} \in [Q]$,
${\sf N}_{{\bf q}'}(\widehat{q})
 = 
 {\sf N}_{{\bf q}}(\widehat{q})$.
We prove: 

\begin{lemma} 
\label{no earlier inversion}
The earliest inversion witness 
in ${\bf q}'$ is either $i$
or some player $\widehat{i} > i$.
\end{lemma}

\begin{proof}
Assume, by way of contradiction,
that the earliest inversion witness in ${\bf q}'$
is a player $j < i$.
Since the earliest inversion witness in ${\bf q}$ is $i$,
$j$ is not an inversion witness in ${\bf q}$.
Let $\widehat{q}$ be the quality chosen by $j$
in ${\bf q}$ and ${\bf q}'$.
Since players other than $i$ and $k$
do not change qualities in ${\bf q}'$,
$j$
makes an inversion pair with either $i$ or $k$
in ${\bf q}'$. 
There are two cases.
\begin{itemize}

\item
\underline{
$j$ makes an inversion pair with $i$
in ${\bf q}'$:}
Since $i$ chooses quality $q_{k}$ in ${\bf q}'$, 
it follows that $\widehat{q} > q'$.
Since $k > j$ and $k$ chooses quality $q_{k}$ in 
${\bf q}$,
this implies that
$j$ and $k$
make an inversion pair in ${\bf q}$.

\item
\underline{$j$ makes an inversion pair with $k$
in ${\bf q}'$:}
Since $k$ chooses quality $q_{i}$ in ${\bf q}'$,
it follows that $\widehat{q} > q_{i}$.
Since $i > j$ and $i$ chooses quality $q_{i}$
in ${\bf q}$,
this implies that $j$ and $i$
make an inversion pair in ${\bf q}$.
\end{itemize}
In either case,
since $i > j$,
$i$ is not the earliest witness of inversion
in ${\bf q}$.
A contradiction.
\end{proof}

\noindent
\remove{Since $\Lambda$ is supermodular,
$s_{i} \geq s_{k}$ and 
${\sf f}_{q} > {\sf f}_{q'}$,
\begin{eqnarray*}
\Lambda (s_{i}, {\sf f}_{q'})
 +
 \Lambda (s_{k}, {\sf f}_{q})
& \leq &
 \Lambda (\max \{ s_{i}, s_{k} \},
          \max \{ {\sf f}_{q}, {\sf f}_{q'} \})
 +
 \Lambda (\min \{ s_{i}, s_{k} \},
          \min \{ {\sf f}_{q}, {\sf f}_{q'} \}) \\
& = &
 \Lambda (s_{i}, {\sf f}_{q})
 +
 \Lambda (s_{k}, {\sf f}_{q'})\, ,\ \mbox{or} \\
         \Lambda (s_{i}, {\sf f}_{q})
         -
         \Lambda (s_{i}, {\sf f}_{q'})
& \geq & \Lambda (s_{k}, {\sf f}_{q})
         -
         \Lambda (s_{k}, {\sf f}_{q'})\, . 	      \hspace{6.0cm}
         \mbox{{\sf (C.3)}}
\end{eqnarray*}}
We continue to prove:

\begin{lemma}
${\sf N}_{{\bf q}'}$ is a pure Nash equilibrium
if and only if
${\sf N}_{{\bf q}}$ is.
\end{lemma}

\begin{proof}
We consider the following cases:
\begin{enumerate}

\item
Player $i$ does not want to switch to quality $q \neq q_{k}$
if and only if
{
\small
\begin{eqnarray*}
&  & {\sf P}_{i}(i, ({\sf N}_{{\bf q}}(1), \ldots,
                     {\sf N}_{{\bf q}}(q_{i}), \ldots,
                   {\sf N}_{{\bf q}}(q_{k}), \ldots,
                   {\sf N}_{{\bf q}}(q), \ldots,
                   {\sf N}_{{\bf q}}(Q)))
   -
   \Lambda (s_{i}, {\sf f}_{q_{k}}) \\
& \geq &
  {\sf P}_{i}(i, ({\sf N}_{{\bf q}}(1), \ldots,
                  {\sf N}_{{\bf q}}(q_{i}), \ldots
                  {\sf N}_{{\bf q}}(q_{k}) - 1, \ldots,
                  {\sf N}_{{\bf q}}(q) + 1, \ldots,
                  {\sf N}_{{\bf q}}(Q)))
  -
  \Lambda (s_{i}, {\sf f}_{q})
\end{eqnarray*}           
}
or
{
\small
\begin{eqnarray*}
      \Lambda (s_{i}, {\sf f}_{q_{k}})
      -
      \Lambda (s_{i}, {\sf f}_{q})
&\leq & {\sf P}_{i}(i, ({\sf N}_{{\bf q}}(1), \ldots,
                        {\sf N}_{{\bf q}}(q_{i}), \ldots,
                        {\sf N}_{{\bf q}}(q_{k}), \ldots,
                        {\sf N}_{{\bf q}}(q), \ldots,
                        {\sf N}_{{\bf q}}(Q)))
        - \\
&     & {\sf P}_{i}(i, ({\sf N}_{{\bf q}}(1), \ldots,
                   {\sf N}_{{\bf q}}(q_{i}), \ldots,
                   {\sf N}_{{\bf q}}(q_{k})-1, \ldots,
                   {\sf N}_{{\bf q}}(q)+1, \ldots,
                   {\sf N}_{{\bf q}}(Q)))\, . {\sf (C.4)}
\end{eqnarray*}
}

\item
Player $k$ does not want to switch to quality $q \neq q_{i}$
if and only if
{
\small
\begin{eqnarray*}
&  & {\sf P}_{k}(k, ({\sf N}_{{\bf q}}(1), \ldots,
                   {\sf N}_{{\bf q}}(q_i), \ldots,
                   {\sf N}_{{\bf q}}(q_{k}), \ldots,
                   {\sf N}_{{\bf q}}(q), \ldots,
                   {\sf N}_{{\bf q}}(Q)))
   -
   \Lambda (s_{k}, {\sf f}_{q_{i}}) \\
& \geq &
  {\sf P}_{k}(k, ({\sf N}_{{\bf q}}(1), \ldots,
                  {\sf N}_{{\bf q}}(q_{i}) -1, \ldots,
                  {\sf N}_{{\bf q}}(q_{k}), \ldots,
                  {\sf N}_{{\bf q}}(q) + 1, \ldots,
                  {\sf N}_{{\bf q}}(Q)))
  -
  \Lambda (s_{k}, {\sf f}_{q})
\end{eqnarray*}           
}
or
{
\small
\begin{eqnarray*}
      \Lambda (s_{k}, {\sf f}_{q_{i}})
      -
      \Lambda (s_{k}, {\sf f}_{q})
&\leq & {\sf P}_{k}(k, ({\sf N}_{{\bf q}}(1), \ldots,
                        {\sf N}_{{\bf q}}(q_{i}), \ldots,
                        {\sf N}_{{\bf q}}(q_{k}), \ldots,
                        {\sf N}_{{\bf q}}(q), \ldots,
                        {\sf N}_{{\bf q}}(Q)))
        - \\
&     & {\sf P}_{k}(k, ({\sf N}_{{\bf q}}(1), \ldots,
                   {\sf N}_{{\bf q}}(q_{i})-1, \ldots,
                   {\sf N}_{{\bf q}}(q_{k}), \ldots,
                   {\sf N}_{{\bf q}}(q)+1, \ldots,
                   {\sf N}_{{\bf q}}(Q)))\, .~~{\sf (C.5)}
\end{eqnarray*}
}

\item
Player $\iota$ does not want to switch to quality
$q_{\kappa} \in [Q] \setminus \{ q_{\iota} \}$
in ${\bf q}'$
if and only if
{
\small
\begin{eqnarray*}
&  & {\sf P}_{\iota}(\iota, ({\sf N}_{{\bf q}}(1), \ldots,
                   {\sf N}_{{\bf q}}(q_{\iota}), \ldots,
                   {\sf N}_{{\bf q}}(q_{\kappa}), \ldots,
                   {\sf N}_{{\bf q}}(Q)))
   -
   \Lambda (s_{\iota}, {\sf f}_{q_{\iota}}) \\
& \geq &
  {\sf P}_{\iota}(\iota, ({\sf N}_{{\bf q}}(1), \ldots,
                  {\sf N}_{{\bf q}}(q_{\iota}) - 1, \ldots,
                  {\sf N}_{{\bf q}}(q_{\kappa}) + 1, \ldots,
                  {\sf N}_{{\bf q}}(Q)))
   -
  \Lambda (s_{\iota}, {\sf f}_{q_{\kappa}})
\end{eqnarray*}           
}
or
{
\small
\begin{eqnarray*}
      \Lambda (s_{\iota}, {\sf f}_{q_{\iota}})
      -
      \Lambda (s_{\iota}, {\sf f}_{q_{\kappa}})
&\leq & {\sf P}_{\iota}(\iota, ({\sf N}_{{\bf q}}(1), \ldots,
                        {\sf N}_{{\bf q}}(q_{\iota}), \ldots,
                        {\sf N}_{{\bf q}}(q_{\kappa}), \ldots,
                        {\sf N}_{{\bf q}}(Q)))
        - \\
&     & {\sf P}_{\iota}(\iota, ({\sf N}_{{\bf q}}(1), \ldots,
                   {\sf N}_{{\bf q}}(q_{\iota}) - 1, \ldots,
                   {\sf N}_{{\bf q}}(q_{\kappa}) + 1, \ldots,
                   {\sf N}_{{\bf q}}(Q)))\, .~~{\sf (C.6)}
\end{eqnarray*}
}

\end{enumerate}

\noindent
Hence, we conclude:
\begin{enumerate}

\item
From the rewriting of the inequality for player $i$
in Definition~\ref{combinatorial condition 1},
\remove{{
\small
\begin{eqnarray*}
& &
{\sf P}_{i}(i, ({\sf N}_{{\bf q}}(1), \ldots,
                   {\sf N}_{{\bf q}}(q_{i})-1, \ldots,
                   {\sf N}_{{\bf q}}'(q_{k}), \ldots,
                   {\sf N}_{{\bf q}}'(q)+1, \ldots,
                   {\sf N}_{{\bf q}}'(Q)))
- \\
& &
 {\sf P}_{i}(i, ({\sf N}_{{\bf q}}(1), 
             \ldots, 
             {\sf N}_{{\bf q}}(q_{i}), 
             \ldots, 
             {\sf N}_{{\bf q}}(q_{k}), \ldots,
             {\sf N}_{{\bf q}}(q), 
             \ldots, 
             {\sf N}_{{\bf q}}(Q)))	\\
& \leq & 
{\sf P}_{i}(i, ({\sf N}_{{\bf q}'}(1), \ldots,
                        {\sf N}_{{\bf q}}(q_{i}), \ldots,
                        {\sf N}_{{\bf q}'}(q_{k}), \ldots,
                        {\sf N}_{{\bf q}'}(q), \ldots,
                        {\sf N}_{{\bf q}'}(Q)))
- \\
& & {\sf P}_{i}(i, ({\sf N}_{{\bf q}}(1), \ldots,
                {\sf N}_{{\bf q}}(q_{i}), \ldots,
                {\sf N}_{{\bf q}}(q_{k})-1, \ldots,
                {\sf N}_{{\bf q}}(q)+1, \ldots,
                {\sf N}_{{\bf q}}(Q)))
\end{eqnarray*}
}
if and only if}
{
\small
\begin{eqnarray*}
& &
  {\sf P}_{i}(i, ({\sf N}_{{\bf q}'}(1), \ldots,
                        {\sf N}_{{\bf q}}(q_{i})-1, \ldots, 
                        {\sf N}_{{\bf q}}(q_{k}), \ldots,
                        {\sf N}_{{\bf q}}(q)+1, \ldots,
                        {\sf N}_{{\bf q}}(Q)))
        - \\
& &        
   {\sf P}_{i}(i, ({\sf N}_{{\bf q}}(1), \ldots,
                   {\sf N}_{{\bf q}}(q_{i}), \ldots,
                   {\sf N}_{{\bf q}}(q_{k}), \ldots,
                   {\sf N}_{{\bf q}}(q), \ldots,
                   {\sf N}_{{\bf q}}(Q)))\, \\
& \leq &
{\sf P}_{i}(i, ({\sf N}_{{\bf q}}(1),
                \ldots
                {\sf N}_{{\bf q}}(q_{i}),
                \ldots,
                {\sf N}_{{\bf q}}(q_{k}), \ldots,
                {\sf N}_{{\bf q}}(q),
                \ldots,
                {\sf N}_{{\bf q}}(Q))) - \\
& &
{\sf P}_{i}(i, ({\sf N}_{{\bf q}}(1), 
             \ldots, 
             {\sf N}_{{\bf q}}(q_{i}), 
             \ldots,
             {\sf N}_{{\bf q}}(q_{k})-1, 
             \ldots, 
             {\sf N}_{{\bf q}}(q)+1, 
             \ldots, 
             {\sf N}_{{\bf q}}(Q)))\, .
\end{eqnarray*}
}
Hence,
{
\small
\begin{eqnarray*}
& &
  {\sf P}_{i}(i, ({\sf N}_{{\bf q}}(1), \ldots,
                        {\sf N}_{{\bf q}}(q_{i})-1, \ldots,
                        {\sf N}_{{\bf q}}(q_{k}), \ldots,
                        {\sf N}_{{\bf q}}(q) + 1, \ldots,
                        {\sf N}_{{\bf q}}(Q)))
        - \\
& &        
   {\sf P}_{i}(i, ({\sf N}_{{\bf q}}(1), \ldots,
                   {\sf N}_{{\bf q}}(q_i), \ldots,
                   {\sf N}_{{\bf q}}(q_{k}), \ldots,
                   {\sf N}_{{\bf q}}(q), \ldots,
                   {\sf N}_{{\bf q}}(Q)))\, \\
& \leq &
\Lambda (s_{i}, {\sf f}_{q})
-
\Lambda (s_{i}, {\sf f}_{q_{i}}),
\Lambda (s_{i}, {\sf f}_{q_{k}})
-
\Lambda (s_{i}, {\sf f}_{q})\\
& \leq &
{\sf P}_{i}(i, ({\sf N}_{{\bf q}}(1),
                \ldots
                {\sf N}_{{\bf q}}(q_{i}),
                \ldots,
                {\sf N}_{{\bf q}}(q_{k}),
                \ldots,
                {\sf N}_{{\bf q}}(q),
                \ldots,
                {\sf N}_{{\bf q}}(Q))) - \\
& &
{\sf P}_{i}(i, ({\sf N}_{{\bf q}}(1), 
             \ldots, 
             {\sf N}_{{\bf q}}(q_{i}), 
             \ldots, 
             {\sf N}_{{\bf q}}(q_{k})-1,
             \ldots,
             {\sf N}_{{\bf q}}(q) + 1, 
             \ldots, 
             {\sf N}_{{\bf q}}(Q)))	
\end{eqnarray*}
}
if and only if
both {\sf (C.1)} and {\sf (C.4)} hold.

\item
From the rewriting of the inequality for player $k$
in Definition~\ref{combinatorial condition 1},
\remove{{
\small
\begin{eqnarray*}
& &
{\sf P}_{k}(k, ({\sf N}_{{\bf q}}(1), \ldots,
                   {\sf N}_{{\bf q}}(q_{i}), \ldots,
                   {\sf N}_{{\bf q}}'(q_{k}), \ldots,
                   {\sf N}_{{\bf q}}'(q)+1, \ldots,
                   {\sf N}_{{\bf q}}'(Q)))
- \\
& &
 {\sf P}_{i}(i, ({\sf N}_{{\bf q}}(1), 
             \ldots, 
             {\sf N}_{{\bf q}}(q_{i}), 
             \ldots, 
             {\sf N}_{{\bf q}}(q_{k}), \ldots,
             {\sf N}_{{\bf q}}(q), 
             \ldots, 
             {\sf N}_{{\bf q}}(Q)))	\\
& \leq & 
{\sf P}_{i}(i, ({\sf N}_{{\bf q}'}(1), \ldots,
                        {\sf N}_{{\bf q}}(q_{i}), \ldots,
                        {\sf N}_{{\bf q}'}(q_{k}), \ldots,
                        {\sf N}_{{\bf q}'}(q), \ldots,
                        {\sf N}_{{\bf q}'}(Q)))
- \\
& & {\sf P}_{i}(i, ({\sf N}_{{\bf q}}(1), \ldots,
                {\sf N}_{{\bf q}}(q_{i}), \ldots,
                {\sf N}_{{\bf q}}(q_{k})-1, \ldots,
                {\sf N}_{{\bf q}}(q)+1, \ldots,
                {\sf N}_{{\bf q}}(Q)))
\end{eqnarray*}
}
if and only if}
{
\small
\begin{eqnarray*}
& &
  {\sf P}_{k}(k, ({\sf N}_{{\bf q}'}(1), \ldots,
                        {\sf N}_{{\bf q}}(q_{i}), \ldots, 
                        {\sf N}_{{\bf q}}(q_{k})-1, \ldots,
                        {\sf N}_{{\bf q}}(q)+1, \ldots,
                        {\sf N}_{{\bf q}}(Q)))
        - \\
& &        
   {\sf P}_{k}(k, ({\sf N}_{{\bf q}}(1), \ldots,
                   {\sf N}_{{\bf q}}(q_{i}), \ldots,
                   {\sf N}_{{\bf q}}(q_{k}), \ldots,
                   {\sf N}_{{\bf q}}(q), \ldots,
                   {\sf N}_{{\bf q}}(Q)))\, \\
& \leq &
{\sf P}_{k}(k, ({\sf N}_{{\bf q}}(1),
                \ldots
                {\sf N}_{{\bf q}}(q_{i}),
                \ldots,
                {\sf N}_{{\bf q}}(q_{k}), \ldots,
                {\sf N}_{{\bf q}}(q),
                \ldots,
                {\sf N}_{{\bf q}}(Q))) - \\
& &
{\sf P}_{k}(k, ({\sf N}_{{\bf q}}(1), 
             \ldots, 
             {\sf N}_{{\bf q}}(q_{i})-1, 
             \ldots,
             {\sf N}_{{\bf q}}(q_{k}), 
             \ldots, 
             {\sf N}_{{\bf q}}(q)+1, 
             \ldots, 
             {\sf N}_{{\bf q}}(Q)))\, .
\end{eqnarray*}
}
Hence,
{
\small
\begin{eqnarray*}
& &
  {\sf P}_{k}(k, ({\sf N}_{{\bf q}}(1), \ldots,
                        {\sf N}_{{\bf q}}(q_{i}), \ldots,
                        {\sf N}_{{\bf q}}(q_{k})-1, \ldots,
                        {\sf N}_{{\bf q}}(q) + 1, \ldots,
                        {\sf N}_{{\bf q}}(Q)))
        - \\
& &        
   {\sf P}_{i}(i, ({\sf N}_{{\bf q}}(1), \ldots,
                   {\sf N}_{{\bf q}}(q_i), \ldots,
                   {\sf N}_{{\bf q}}(q_{k}), \ldots,
                   {\sf N}_{{\bf q}}(q), \ldots,
                   {\sf N}_{{\bf q}}(Q)))\, \\
& \leq &
\Lambda (s_{k}, {\sf f}_{q})
-
\Lambda (s_{k}, {\sf f}_{q_{i}}),
\Lambda (s_{k}, {\sf f}_{q_{k}})
-
\Lambda (s_{k}, {\sf f}_{q})\\
& \leq &
{\sf P}_{k}(k, ({\sf N}_{{\bf q}}(1),
                \ldots
                {\sf N}_{{\bf q}}(q_{i}),
                \ldots,
                {\sf N}_{{\bf q}}(q_{k}),
                \ldots,
                {\sf N}_{{\bf q}}(q),
                \ldots,
                {\sf N}_{{\bf q}}(Q))) - \\
& &
{\sf P}_{k}(k, ({\sf N}_{{\bf q}}(1), 
             \ldots, 
             {\sf N}_{{\bf q}}(q_{i})-1, 
             \ldots, 
             {\sf N}_{{\bf q}}(q_{k}),
             \ldots,
             {\sf N}_{{\bf q}}(q) + 1, 
             \ldots, 
             {\sf N}_{{\bf q}}(Q)))	
\end{eqnarray*}
}
if and only if
both {\sf (C.2)} and {\sf (C.5)} hold.

\item
Since {\sf (C.3)} and {\sf (C.6)} are identical,
it follows that
player $\iota$ does not want to switch
to a quality $q_{\kappa} \neq q_{\iota}$
in ${\bf q}$ 
if and only if
she does not want to switch
to the quality $q_{\kappa}$ in ${\bf q}'$.

\end{enumerate}
The conclusions imply that no player
wants to switch qualities in ${\bf q}$
if and only if
she does not want to switch qualities in ${\bf q}'$.
The claim follows.
\remove{\textcolor{red}{do not want to switch in ${\bf q}'$
since 
{\it (i)}
they did not want to switch in ${\bf x}$,
{\it (ii)}
they choose the same qualities as in ${\bf x}$
and {\it (iii)}
${\sf N}_{{\bf x}'}(\widehat{q}) = 
 {\sf N}_{{\bf x}}(\widehat{q})$ 
for all qualities $\widehat{q} \in [Q]$.
Hence,
${\bf x}'$ is a pure Nash equilibrium.}}
\end{proof}

\noindent
Now the earliest inversion witness,
if any,
in ${\bf q}'$ is either $i$,
the earliest witness of inversion in ${\bf q}$,
making an inversion pair with a player $\widehat{k} > k$,
or greater than $i$. 
It follows inductively that
a pure Nash equilibrium exists
if and only if
a contiguous pure Nash equilibrium exists.
\end{proof}

\noindent
By Proposition~\ref{transformation into contiguous},
it suffices to search over contiguous load vectors.
\remove{
Recall that for any player $i$ assigned to quality $q \in [Q]$,
$c_{i} 
 = s_{i} {\sf f}_{q} - \beta\, 
   \frac{\textstyle {\sf f}_{q}}
        {\textstyle \sum_{\check{q} \in [Q]} 
        x_{\check{q}} {\sf f}_{\check{q}}}$.                            
}
Fix a load vector ${\sf N}_{{\bf q}}$
and a quality $q \in [Q]$
such that ${\sf Players}_{{\bf q}}(q) \neq \emptyset$. 
No player choosing quality $q$ wants to switch 
to the quality $q' \neq q$
if and only if
for all players $i \in {\sf Players}_{{\bf q}}(q)$,
\remove{\begin{eqnarray*}
         s_{i} {\sf f}_{q} - \beta\, 
         \frac{\textstyle {\sf f}_{q}}
              {\textstyle \sum_{\check{q} \in [Q]}
                          x_{\check{q}} {\sf f}_{\check{q}}}
& \leq & s_{i} {\sf f}_{q'} 
         - \beta\, \frac{\textstyle {\sf f}_{q'}}
                        {\textstyle \sum_{\check{q} \in [Q] \mid \check{q} \neq q, q'} {\sf f}_{\check{q}} x_{\check{q}} 
            + {\sf f}_{q} (x_{q} -1) 
            + {\sf f}_{q'} (x_{q'} +1)}
\end{eqnarray*}}
\begin{eqnarray*}
         {\sf P}_{i}(i, {\sf N}_{{\bf q}}) -
         \Lambda (s_{i}, {\sf f}_{q})
& \geq & {\sf P}_{i}(i,
                     ({\sf N}_{{\bf q}}(1),
                     \ldots,
                     {\sf N}_{{\bf q}}(q) - 1,
                     \ldots,
                     {\sf N}_{{\bf q}}(q') + 1,  	                 \ldots,
                     {\sf N}_{{\bf q}}(Q))) -
         \Lambda (s_{i}, {\sf f}_{q'})\, 
\end{eqnarray*}
or
\begin{eqnarray*}
         \Lambda (s_{i}, {\sf f}_{q})
         -
         \Lambda (s_{i}, {\sf f}_{q'})
& \leq & {\sf P}_{i}(i, {\sf N}_{{\bf q}})
         -
         {\sf P}_{i}(i,
                     ({\sf N}_{{\bf q}}(1),
                     \ldots,
                     {\sf N}_{{\bf q}}(q) - 1,
                     \ldots,
                     {\sf N}_{{\bf q}}(q') + 1,  	                 \ldots,
                     {\sf N}_{{\bf q}}(Q)))\, . {\sf (C.7)}
\end{eqnarray*}
Since ${\sf P}$ is player-specific,
${\sf P}_{i}(i, {\sf N}_{{\bf q}})$
and ${\sf P}_{i}(i, ({\sf N}_{{\bf q}}(1),
                 \ldots,
                 {\sf N}_{{\bf q}}(q) - 1,
                 \ldots,
                 {\sf N}_{{\bf q}}(q') + 1,
                 \ldots,
                 {\sf N}_{{\bf q}}(Q)))$
are not constant
over all players
choosing quality $q$ in ${\sf N}_{{\bf q}}$
and switching to quality $q'$ in
$({\sf N}_{{\bf q}}(1),
  \ldots,
  {\sf N}_{{\bf q}}(q) - 1,
  \ldots,
  {\sf N}_{{\bf q}}(q') + 1,
  \ldots,
  {\sf N}_{{\bf q}}(Q))$,
respectively.  
Hence,
no player choosing quality $q \in [Q]$
wants to switch to a quality $q' \neq q$
if and only if
{\sf (C.4)} holds for all players
$i \in {\sf Players}_{{\bf q}}(q)$.

To compute a pure Nash equilibrium,
we enumerate all contiguous load vectors
${\sf N}_{{\bf q}} = \langle {\sf N}_{{\bf q}}(1), 
                   {\sf N}_{{\bf q}}(2), 
                   \ldots,
                   {\sf N}_{{\bf q}}(Q)
           \rangle$,
searching for one that satisfies 
{\sf (C.7)},
for each quality $q \in [Q]$
and 
for all players
$i \in {\sf Players}_{{\bf q}}(q)$;
clearly,
there are $\binom{\textstyle n+Q-1}{\textstyle Q-1}$
contiguous load vectors (cf.~\cite[Section 2.6]{C02}). 
\remove{
\begin{itemize}
	
\item
For player-invariant payments,
checking {\sf (C.5)} 
for a quality $q \in [Q]$
entails the computation of the minimum of a function
on a set of size ${\sf N}_{{\bf x}}(q)$;
computation of the minima for all qualities $q \in [Q]$
takes time $\sum_{q \in [Q]}
             \Theta ({\sf N}_{{\bf x}}(q))
           =
           \Theta \left( \sum_{q \in [Q]}
                          {\sf N}_{{\bf x}}(q)
                  \right)
           =
           \Theta (n)$.
Thus, 
the total time is
$\binom{\textstyle n+Q-1}{\textstyle Q-1}
 \cdot
 \left( \Theta(n)
        +
        \Theta (Q^{2})
 \right)
 =
 \Theta \left( \max \{ n, Q^{2} \}
               \cdot
               \binom{\textstyle n+Q-1}{\textstyle Q-1}
        \right)$.
           
\item           
\textcolor{violet}
}
For a player-specific payment function,
checking {\sf (C.7)}
for a quality $q \in [Q]$
entails no minimum computation
but must be repeated $n$ times
for all players $i \in [n]$;            
checking that the inequality holds
for a particular $q' \neq q$ takes time $\Theta (1)$,
so checking that it holds for all qualities $q' \neq q$
takes time $\Theta (Q)$,
and checking that it holds for all $q \in [Q]$
takes time $\Theta (Q^{2})$.
Thus, 
the total time is
$\Theta \left( n
               \cdot
               Q^{2}
               \cdot
               \binom{\textstyle n+Q-1}{\textstyle Q-1}
        \right)$. 
For constant $Q$,
this is a polynomial $\Theta \left( n^{Q} \right)$ 
algorithm.

\remove{Clearly,
there are 
$\binom{\textstyle n}{Q-1}$ 
contiguous profiles to enumerate,
and for each we have to check $2Q-2$ conditions,
involving together the computation of a minimum and a maximum
over two sets of total size $Q-1$.
(For all qualities $q \in [Q]$ except $1$ and $Q$,
for which we have to check only one,
we have to check two conditions.)
Hence,
we have a $\Theta \left( Q^{2} \binom{\textstyle n}{Q-1} \right)$ algorithm.}
By Proposition~\ref{transformation into contiguous},
a contiguous load vector satisfying {\sf (C.7)}
for each quality $q \in [Q]$ exists
if and only if
it will be found by the algorithm
enumerating all contiguous load vectors.
Hence,
the algorithm solves
{\sc $\exists$PNE with Player-Specific Payments}.
\end{proof}

\subsection{Player-Invariant Payments}

Recall that
a player-invariant payment function ${\sf P}_{i}({\bf q})$
can be represented by a two-argument payment function
${\sf P}_{i}(q, {\bf q}_{-i})$,
where $q \in [Q]$
and ${\bf q}_{-i}$
is a partial quality vector,
for some player $i \in [n]$.
In correspondence to three-discrete-concave player-specific payments,
we define:

\begin{definition}
\label{combinatorial condition 2}
A player-invariant payment function
${\sf P}$
is {\it three-discrete-concave} if
for every player $i \in [n]$,
for every load vector ${\sf N}_{{\bf q}}$
and for every triple of qualities
$q_{i}$, $q_{k}$, $q \in [Q]$,
{
\small
\begin{eqnarray*}
& &
{\sf P}_{i}(q, ({\sf N}_{{\bf q}}(1), \ldots,
                   {\sf N}_{{\bf q}}(q_{i}), \ldots,
                   {\sf N}_{{\bf q}}(q_{k})-1, \ldots,
                   {\sf N}_{{\bf q}}(q)+1, \ldots,
                   {\sf N}_{{\bf q}}(Q)))
+ \\
& &
 {\sf P}_{i}(q, ({\sf N}_{{\bf q}}(1), 
             \ldots, 
             {\sf N}_{{\bf q}}(q_{i}) - 1, 
             \ldots, 
             {\sf N}_{{\bf q}}(q_{k}),
             \ldots,
             {\sf N}_{{\bf q}}(q) + 1, 
             \ldots, 
             {\sf N}_{{\bf q}}(Q)))	\\
& \leq & 
2\,
{\sf P}_{i}(q_{i}, ({\sf N}_{{\bf q}}(1), \ldots,
                        {\sf N}_{{\bf q}}(q_{i}), \ldots, 
                        {\sf N}_{{\bf q}}(q_{k}), \ldots,
                        {\sf N}_{{\bf q}}(q), \ldots,
                        {\sf N}_{{\bf q}}(Q)))\, .
\end{eqnarray*}
}
\end{definition}

\noindent
In correspondence to Proposition~\ref{transformation into contiguous},
we prove a Contigufication Lemma
for three-discrete-concave player-invariant payment functions:

\begin{proposition}[{\bf Contigufication Lemma for Player-Invariant Payments}]
\label{transformation into contiguous 2}
For three-discrete-concave player-invariant payments,  
any pair of
{\it (i)}
a pure Nash equilibrium 
${\sf N}_{{\bf q}}
 = \langle {\sf N}_{{\bf q}}(1),
           \ldots, 
           {\sf N}_{{\bf q}}(Q)
   \rangle$
and 
{\it (ii)}
player sets ${\sf Players}_{{\bf q}}(q)$
for each quality $q \in [Q]$,
can be transformed into a contiguous pure Nash equilibrium.
\end{proposition}

\noindent
By Proposition~\ref{transformation into contiguous 2},
it suffices to search over contiguous load vectors.
Fix a load vector
${\sf N}_{{\bf q}}$
and a quality $q \in [Q]$
such that
${\sf Players}_{{\bf q}}(q) \neq \emptyset$.
No player choosing quality $q$
wants to switch to the quality $q' \neq q$
if and only if
for all players $i \in {\sf Players}_{{\bf q}}(q)$,
\begin{eqnarray*}
         {\sf P}_{i}(q, {\sf N}_{{\bf q}}) -
         \Lambda (s_{i}, {\sf f}_{q})
& \geq & {\sf P}_{i}(q',
                     ({\sf N}_{{\bf q}}(1),
                     \ldots,
                     {\sf N}_{{\bf q}}(q) - 1,
                     \ldots,
                     {\sf N}_{{\bf q}}(q') + 1,  	                 \ldots,
                     {\sf N}_{{\bf q}}(Q))) -
         \Lambda (s_{i}, {\sf f}_{q'})\, 
\end{eqnarray*}
or
\begin{eqnarray*}
         \Lambda (s_{i}, {\sf f}_{q})
         -
         \Lambda (s_{i}, {\sf f}_{q'})
& \leq & {\sf P}_{i}(q, {\sf N}_{{\bf q}})
         -
         {\sf P}_{i}(q',
                     ({\sf N}_{{\bf q}}(1),
                     \ldots,
                     {\sf N}_{{\bf q}}(q) - 1,
                     \ldots,
                     {\sf N}_{{\bf q}}(q') + 1,  	                 \ldots,
                     {\sf N}_{{\bf q}}(Q)))\, . {\sf (C.8)}
\end{eqnarray*}

\noindent
Since ${\sf P}$ is player-invariant,
${\sf P}_{i}(q, {\sf N}_{{\bf q}})$
and ${\sf P}_{i}(q', ({\sf N}_{{\bf q}}(1),
                 \ldots,
                 {\sf N}_{{\bf q}}(q) - 1,
                 \ldots,
                 {\sf N}_{{\bf q}}(q') + 1,
                 \ldots,
                 {\sf N}_{{\bf q}}(Q)))$
are constant
over all players
choosing quality $q$ in ${\sf N}_{{\bf q}}$
and switching to quality $q'$ in
$({\sf N}_{{\bf q}}(1),
  \ldots,
  {\sf N}_{{\bf q}}(q) - 1,
  \ldots,
  {\sf N}_{{\bf q}}(q') + 1,
  \ldots,
  {\sf N}_{{\bf q}}(Q))$,
respectively.
Hence,
no player $\widehat{i}$ choosing quality $q \in [Q]$
wants to switch to a quality $q' \neq q$
if and only if
{\sf (C.8)} holds 
for each quality $q' \neq q$,
where 
$\widehat{i} \in {\sf Players}_{{\bf q}}(q)$
is arbitrarily chosen.

To compute a pure Nash equilibrium,
we enumerate all contiguous load vectors
${\sf N}_{{\bf q}} = \langle {\sf N}_{{\bf q}}(1), 
                   {\sf N}_{{\bf q}}(2), 
                   \ldots,
                   {\sf N}_{{\bf q}}(Q)
           \rangle$,
searching for one that satisfies {\sf (C.8)},
for each quality $q \in [Q]$
and 
for a player
$\widehat{i} \in {\sf Players}_{{\bf q}}(q)$;
clearly,
there are $\binom{\textstyle n+Q-1}{\textstyle Q-1}$
contiguous load vectors (cf.~\cite[Section 2.6]{C02}. 
For player-invariant payments,
checking {\sf (C.8)} 
for a quality $q \in [Q]$
entails the computation of the minimum of a function
on a set of size ${\sf N}_{{\bf q}}(q)$;
computation of the minima for all qualities $q \in [Q]$
takes time $\sum_{q \in [Q]}
             \Theta ({\sf N}_{{\bf q}}(q))
           =
           \Theta \left( \sum_{q \in [Q]}
                          {\sf N}_{{\bf q}}(q)
                  \right)
           =
           \Theta (n)$.
Thus, 
the total time is
$\binom{\textstyle n+Q-1}{\textstyle Q-1}
 \cdot
 \left( \Theta(n)
        +
        \Theta (Q^{2})
 \right)
 =
 \Theta \left( \max \{ n, Q^{2} \}
               \cdot
               \binom{\textstyle n+Q-1}{\textstyle Q-1}
        \right)$.

By Proposition~\ref{transformation into contiguous},
a contiguous load vector satisfying {\sf (C.8)}
for each quality $q \in [Q]$ exists
if and only if
it will be found by the algorithm
enumerating all contiguous load vectors.                
Hence, it follows:

\begin{theorem}
\label{the most most arbitrary}
There is a 
$\Theta \left( 
        \max \{ n, Q^2 \}
         \cdot
         \binom{\textstyle n+Q-1}{\textstyle Q-1} 
         \right)$ 
algorithm that solves 
{\sc $\exists$PNE with Player-Invariant Payments}
for arbitrary players
and three-discrete-concave player-invariant payments;
for constant $Q$, 
it is a $\Theta (n^{Q})$ 
polynomial algorithm.
\end{theorem}

\begin{framed}
\begin{open problem}
Investigate the possibility of improving the time complexities
of the algorithms in Theorems~\ref{the most arbitrary}
and~\ref{the most most arbitrary}.
For constant $Q$,
this means reducing the exponent $Q$ of $n$.
Assumptions stronger than three-discrete-concavity
on the payments
might be required.
\end{open problem}
\end{framed}

\section{Open Problems and Directions for Further Research}
\label{epilogue}


This work poses far more challenging 
problems and research directions
about the contest game
than it answers. To close we list a few open
research directions.
\begin{enumerate}

\remove{
\item
Study the {\it uniqueness} of pure Nash equilibria;
non-uniqueness would trigger
interesting decision problems.
}


\remove{
\item
Study the computation of pure
Nash equilibria
for other classes of player-invariant payment functions,
such as those mentioned
in Section~\ref{the contest game for crowdsourcing reviews}.
}

\item
Study the computation of {\it mixed} Nash equilibria.
Work in progress confirms the existence
of contest games with 
$Q=3$ and $n=3$
that have only one mixed Nash equilibrium,
which is irrational.
We conjecture that the problem
is ${\mathcal{PPAD}}$-complete
for $n=2$.
\remove{Assume there is only one proposal.
A {\it mixed strategy} of player $i \in [n]$ 
is a vector of probabilities
$\boldsymbol{\pi}_{i}
 = \langle \pi_{i1}, \pi_{i2}, 1 - \pi_{i1} - \pi_{i2} \rangle$;
$\pi_{iq}$ is the probability that player $i$
chooses quality $q \in \{ 1, 2 \}$.
The matrix $\boldsymbol{\Pi} = 
\langle \boldsymbol{\pi}_{1},
        \boldsymbol{\pi}_{2},
        \boldsymbol{\pi}_{3} \rangle$
of the mixed strategies of the players
is called the {\it mixed strategy matrix};
it corresponds to the usual notion of mixed profile.        
}

\remove{
\item
Investigate conditions
on the payment function
and the skill-effort function
for the contest game for crowdsourcing reviews
to be a {\it valid utility game}~\cite{V02}.
Given the {\it general} existence result for pure Nash equilibria
under a player-invariant payment function
in Theorem~\ref{pure existence},
this may open up the road
to upper-bound the Price of Anarchy 
for arbitrary $Q$.
(The particular upper bounds
for proportional allocation in~\cite{BKLO22}
either do not go beyond the case
$Q=3$ under voluntary participation, 
for which they proved
existence of 
pure Nash equilibria~\cite[Proposition 2]{BKLO22},
or go beyond this case
without proving existence first~\cite[Theorem 4]{BKLO22}.)
}

\item
Determine the complexity of computing {\it best-responses}
for the contest game.
We conjecture ${\cal NP}$-hardness;
techniques similar
to those used in~\cite[Section 3]{EGG22} could be useful.

\item
Formulate incomplete information contest games
with discrete strategy spaces
and study their Bayes-Nash equilibria.
Ideas from Bayesian congestion games~\cite{GMT08}
will very likely be helpful.
Study existence and complexity properties
of pure Bayes-Nash equilibria. 

\item
In analogy to weighted congestion games~\cite{M96,R73},
formulate the {\it weighted} contest game
with discrete strategy spaces,
where reviewers have {\it weights,}
and study its pure Nash equilibria.

\remove{
\item
Formulate and study 
{\it malicious Bayesian contest games,}
extending work on malicious Bayesian congestion games~\cite{G08}.
Study (in)existence and complexity properties
of their Bayes-Nash equilibria.
}

\remove{
\item
Incorporate and study issues of interaction, cooperation
and competition both among reviewers and 
among {\it proposers} (cf.~\cite{DLLQ22}).
}

\remove{
\item
Two natural performance metrics
are studied in~\cite{BKLO22}
and bounds on the {\it Price of Anarchy} 
are presented for them.                           
}

\end{enumerate}

\noindent
{\bf Acknowledgements.}
We would like to thank the anonymous referees
to a previous version of the paper
for some very insightful comments they offered.

\end{document}